\documentclass[10pt,twoside,web]{ieeecolor}
\usepackage{generic}
\usepackage{tensor}
\usepackage{bbold}
\usepackage[noadjust]{cite}
\usepackage{enumerate}
\usepackage{graphicx}
\usepackage{amsmath,amssymb,amsfonts}
\usepackage{lipsum}
\usepackage{abraces}
\usepackage{tcolorbox}
\usepackage{hyperref}
\usepackage{xcolor}

\usepackage{color}
\usepackage{multirow}
\usepackage{mathrsfs}
\usepackage{lineno,hyperref}
\usepackage{dsfont}
\usepackage{mathtools}
\usepackage{amssymb}
\usepackage{cite}
\usepackage{float}
\usepackage{amsmath,bm,bbm}
\usepackage{color}
\usepackage{subfigure}
\usepackage{algorithm,algorithmic,setspace}
\usepackage{soul}

\newtheorem{remark}{{\it Remark}\rm}
\newtheorem{lemma}{Lemma}
\newtheorem{definition}{Definition}
\newtheorem{theorem}{Theorem}
\newtheorem{assumption}{Assumption}

\def\begmat#1{\begin{bmatrix}#1\end{bmatrix}}

\def\vbf#1{\boldsymbol{#1} }

\def\cale{{\cal E}}
\def\cali{{\cal I}}
\def\calo{{\cal O}}

\def\calc{{\cal C}}
\def\calt{{\cal T}}
\def\calb{{\cal B}}
\def\calm{{\cal M}}

\def\cale{{\cal E}}

\def\calv{{\cal V}}
\def\bfr{{\vbf r}}
\def\bfq{\vbf{q}}
\def\bfy{{\bf y}}
\def\bfz{{\bf z}}

\def\bfx{{\bf x}}

\def\bfu{{\bf u}}

\def\bfv{{\bf v}}
\def\bfa{{\bf a}}
\def\bfw{\vbf{w}}

\def\chib{\vbf\chi}
\def\ellb{\vbf\ell}
\def\thetab{\vbf{\theta}}
\def\xib{\vbf{\xi}}
\def\Omegab{\vbf\Omega}
\def\phib{ \phi}

\def\phil{{\phi}^{\mathtt{L}}}
\def\pa{\mbox{skew}}

\def\L2e{{\cal L}_{2e}}

\def\rea{\mathbb{R}}

\def\adj{\mbox{adj}}
\def\col{\mbox{col}}
\def\hal{{1 \over 2}}
\def\et{\epsilon_t}

\def\min{{\mbox{min}}}

\def\span{{\mbox{span}}}

\usepackage[prependcaption,colorinlistoftodos]{todonotes}

\def\li{{{\ell_i}}}

\def\calc{{\cal C}}

\def\calb{{\cal B}}
\def\cale{{\cal E}}

\def\calm{{\cal M}}
\def\caln{{\cal N}}
\def\calt{{\cal T}}

\def\calo{{\cal O}}

\def\bq{\mathbb{Q}}
\def\hal{{1 \over 2}}

\def\col{\mbox{col}}

\def\L2{{\cal L}_2}
\def\L2e{{\cal L}_{2e}}

\def\rea{\mathbb{R}}

\def\begequarr{\begin{eqnarray}}
\def\endequarr{\end{eqnarray}}
\def\begequarrs{\begin{eqnarray*}}
\def\endequarrs{\end{eqnarray*}}
\def\begarr{\begin{array}}
\def\endarr{\end{array}}
\def\begequ{\begin{equation}}
\def\endequ{\end{equation}}

\def\begenu{\begin{enumerate}}
\def\begite{\begin{itemize}}
\def\endite{\end{itemize}}
\def\endenu{\end{enumerate}}

\def\lef[{\left[\begin{array}}
\def\rig]{\end{array}\right]}
\def\begcen{\begin{center}}
\def\endcen{\end{center}}

\def\ki{k_{\tt I}}

\def\tr{\mbox{tr}}

\def\QED{\hfill $\blacksquare$}
\def\qed{\hfill $\triangleleft$}

\def\begmat#1{\begin{bmatrix}#1\end{bmatrix}}

\usepackage{color}

\newcommand{\algocomment}[1]{{\hfill\ttfamily\footnotesize\fontdimen2\font=0.3em\textcolor{magenta}{/*#1*/}}}

\def\BibTeX{{\rm B\kern-.05em{\sc i\kern-.025em b}\kern-.08em
    T\kern-.1667em\lower.7ex\hbox{E}\kern-.125emX}}

\begin{document}

\title{
PEBO-SLAM: Observer Design for Visual Inertial SLAM with Convergence Guarantees
}

\author{
Bowen Yi, Chi Jin, Lei Wang, Guodong Shi, Viorela Ila, and Ian R. Manchester
\thanks{This paper was supported by the Australian Research Council (ARC). L. Wang was supported by the National Natural Science Foundation of China (62203386), and Zhejiang Provincial Natural Science Foundation of China (LZ23F030008). {\em (Corresponding author: Lei Wang)}
}
\thanks{B. Yi, G. Shi, V. Ila and I.R. Manchester are with the Australian Centre for Robotics, School of Aerospace, Mechanical and Mechatronic Engineering, The University of Sydney, Sydney, NSW 2006, Australia. (e-mail: bowen.yi@uts.edu.au, \{guodong.shi,~viorela.ila,~ian.manchester\}@sydney.edu.au)}
\thanks{
C. Jin is with the TuSimple Inc., Liangmaqiao Road, Beijing, China. (e-mail:  jinchitrue@163.com) }%
\thanks{
L. Wang is with State Key Laboratory of Industrial Control Technology, Institute of Cyber-Systems and Control, Zhejiang University, Hangzhou 310027, China (e-mail:  lei.wangzju@zju.edu.cn)
}
}

\maketitle
\thispagestyle{empty}

\begin{abstract}
This paper introduces a new linear parameterization to the problem of visual inertial simultaneous localization and mapping (VI-SLAM)---without any approximation---for the case only using information from a single monocular camera and an inertial measurement unit. In this problem set, the system state evolves on the nonlinear manifold $SE(3)\times \rea^{3n}$, on which we design dynamic extensions carefully to generate invariant foliations, such that the problem can be reformulated into online \emph{constant parameter} identification, then interestingly with linear regression models obtained. It demonstrates that VI-SLAM can be translated into a linear least squares problem, in the deterministic sense, \emph{globally} and \emph{exactly}. Based on this observation, we propose a novel SLAM observer, following the recently established parameter estimation-based observer (PEBO) methodology. A notable merit is that the proposed observer enjoys almost global asymptotic stability, requiring neither persistency of excitation nor uniform complete observability, which, however, are widely adopted in most existing works with provable stability but can hardly be assured in many practical scenarios. 
\end{abstract}


%
\section{Introduction}
\label{sec1}
%
\subsection{Literature review}
Simultaneous localization and mapping (SLAM) is a fundamental problem widely studied in the robotic and automation communities, as well as in the field of navigation \cite{MONTHR,DISetal,SUetal,VANetal,ZLOFOR}, see \cite{CADetal} for a recent review. In SLAM, two main aims are concerned to be accomplished concurrently---mapping an unknown environment, and online estimating the pose (i.e. attitude and position) of a mobile robot. Nowadays, SLAM is becoming an significantly important part for unmanned systems in the absence of absolute positioning systems.

The algorithmic approaches at the SLAM back end may generally be classified into two categories, i.e., smoothing and filtering, and the interested reader may refer to \cite{STRetal} for a comprehensive comparison. The former is commonly referred as \emph{maximum a posteriori} (MAP) estimation, or smoothing and mapping (SaM), which can historically date back to the pioneer work \cite{LUetal}, and often be represented using factor graphs. In smoothing methods, state estimation is treated as batch optimization. To be precise, with the aim to estimate the system state $\bfx(k)$ at the current moment $k \in \mathbb{N}_+$, we collect all its past trajectories in an extended state $X:=[\bfx(1)^\top , \ldots, \bfx(k)^\top]^\top$ and formulate SLAM as maximizing the posterior under some Gaussian assumptions, given a set of historical record of measurement $\bar\bfy:=\{\bfy(j), j=1,\ldots,k\}$ as well as the kinematic and observation models. Using the Bayes' theorem and a uniform distribution assumption, the MAP estimation reduces to maximum likelihood estimation, or equivalently a nonlinear least squares problem \cite[Sec. 2]{CADetal}. Compared to filtering, this approach usually has a better accuracy in particular for large-scaled mappings; however, on the other hand, it has a serious issue of scalability, i.e., heavy memory consumption and growing size of the factor graph. Such an issue has been partially solved by utilizing sparsification \cite{ILAetalTRO}, Parallel (out-of-core) SLAM \cite{BOSetal} and incremental smoothing \cite{KAEetal,ILAetal}. In contrast, the filtering approach relies on recursive or incremental algorithms to get the pose and features estimation asymptotically, in which only current pose is in need to be estimated, and thus the dimension of the pose state will not grow boundlessly over time. This school includes extended Kalman filter (EKF)-SLAM and many modern algorithms, e.g. FastSLAM, which combines EKF and particle filtering \cite{MONTHR}. An essential part of these algorithms is their convergence and consistency analysis, the success of which relies on linearization of nonlinear dynamical models \cite{HUADIS}. It sometimes yields satisfactory performance, but may have inconsistency issues when starting from a bad initial guess, which is caused by small domains of attraction in terms of approximation---invoking high nonlinearity of the associated dynamics and output functions.

While new state-of-the-art techniques continue to come up these years, some long-standing technical and theoretical challenges are still open. As reviewed above, a key factor stymieing the performance enhancement in practice may refer to non-convexity in smoothing and locality in filtering; see for example \cite{WANetalAUT} for an attempt of understanding the nonlinearity structure in one-step SLAM. With such a consideration, a natural question openly arises in \cite{HUAetal} where it is stated:
``{\em How far is SLAM from a linear least squares problem?}'' The main purpose of this paper is to give an affirmative answer to the above question in the context of visual inertial SLAM \emph{without} involving any linearization or approximation. Its importance can hardly be overestimated, since, on one hand, it may help us have a better understanding of nonlinearity in SLAM, and on the other hand it provides the possibility to enhance performance in practice. Indeed, many efforts have been made towards this target. In \cite{ZHAetal}, the authors consider local submap joining via solving linear least squares problems, which is an approximation to the optimal full nonlinear least squares SLAM, though non-convex optimization is adopted to obtain submaps. Meanwhile, some very recent works \cite{LOUetal,TANetal,GUEetal}, from the filtering perspective, shows that \emph{robo-centric} SLAM problem, i.e., estimating landmark coordinates in the body-fixed frame, can be \emph{exactly} transformed into state observation of linear time-varying (LTV) systems, via judiciously selecting alternative outputs and systems state. Then, global asymptotic convergence can be achieved using LTV Kalman-Bucy filters, if the robot trajectory guarantees uniform complete observability (UCO). Our paper continues along this line of research as \cite{LOUetal,TANetal}, but with the key difference of focusing on \emph{world-centric} SLAM.

In recent years, the nonlinear control community shows growing interests for SLAM, providing alternative solutions by means of nonlinear observer design. Indeed, the problem may be regarded as state observation of a nonlinear system living on the manifold $SE(3) \times \rea^{3n}$, whose output function depends on specific sensors. On the other hand, nonlinear observer on manifolds is a well-established topic, with special emphasis to matrix Lie groups \cite{IZASAN,LAGetal,MAHetal}. In this paper, we are interested in state observer design for VI-SLAM, a case with only bearing measurement of features and velocities available. Similar problems were recently studied in \cite{VANetal}, in which the authors introduce a constructive observer design by lifting to a new symmetry Lie group $\mathbb{VSLAM}(3)$ in order to make the output function equivariant. Since the model is not strongly differentially observable \cite{BES}, in order to be able to achieve asymptotic stability, some persistency of excitation (PE) conditions are required for the robot trajectory; see for example \cite{LOUetal,TANetal} imposing the UCO-type assumptions. It is well known that UCO of LTV systems is equivalent to PE of some intermediate variables \cite{SASBOD}. Besides, the PE and UCO conditions are also indispensable in some related problems, e.g., locolization using range or direction measurements \cite{HAMSAM}, velocity estimation using normalized measurement \cite{BJOetal} and feature depth observation \cite{DELUCA}. Theoretically, the absence of UCO or PE is yet another source of inaccuracy and inconsistency---apart from nonlinearity and non-convexity as illustrated above---in filtering approaches. Intuitively, these excitation conditions impose \emph{relative motion} between the mobile robot and features, the \emph{uniformity} of which should hold with respect to time. However, such assumptions may not be satisfied in many practical scenarios, such as, robots stopping in specific tasks, and features appearing in camera only during a short interval. Under these circumstances, estimates from the existing SLAM observers fail to converge to their true values; and observer states may even diverge in the presence of measurement noise. In contrast, smoothing approaches only require some ``informative'' assumptions without the need of uniformity of time, since it solves optimization in each step directly rather than in a recursive manner. Overcoming the inconsistency from the lack of sufficient excitation for the filtering approach is another motivation of the paper.

\subsection{Overview and Paper Organization}
\label{sec1b}

There are many different ways to formulate SLAM problems, among which smoothing and filtering/observer are two de facto standard frameworks. In the former, it is formulated as a large batch stochastic optimization problem; and in filtering the problem is addressed as information fusion. In this work, we aim to provide a filtering solution to estimate the entire trajectory (pose) and the map.  

The proposed PEBO-SLAM framework can be distinguished from conventional smoothing approaches (based on nonlinear optimization) and incremental/recursive SLAM filters (e.g. nonlinear observers and several variants of EKFs), i.e., the proposed approach relies on using dynamic extensions and then treating visual inertial SLAM as a linear squares problem without any approximation. After obtaining linear regressor equations (LREs), instead of solving off-line optimization directly we provide a consistent on-line estimator. The PEBO-SLAM framework enjoys three key merits, i.e., low computational burden, less data memory, and (almost) global stability guarantees. Although the algorithm in the paper are tailored for a specific SLAM problem, i.e., three-dimensional visual inertial SLAM, the PEBO-SLAM framework is applicable to many other SLAM problems.

The paper is organized as follows. In the remainder of this section, we will summarize the main contributions, and introduce mathematical notations. In Sections \ref{sec2s} and \ref{sec2} we present some preliminaries of PEBO and the problem formulation studied in the paper, respectively. In Section \ref{sec3}, a novel visual inertial SLAM observer framework will be designed, and then we propose a feasible solution in Section \ref{sec4}. It will be followed by some simulation results in Section \ref{sec5}. Finally, the paper is wrapped up by a brief concluding remark. 

\subsection{Contributions}

We consider the mobile robot moving in three-dimensional space, and our task is to asymptotically estimate the robot pose and feature positions, only using the information of landmark bearings observed from a single monocular camera, and the odometry from an inertial measurement unit (IMU). The constructive nonlinear observer tool which we adopt is the recently proposed \emph{parameter estimation-based observer} (PEBO) \cite{ORTetalscl}; see \cite{ORTetalaut} for a generalized version, and \cite{YIetal} for its geometric interpretation. In the paper, we show that PEBO can be extended from Euclidean space to Lie matrix groups. We call the proposed design framework as PEBO-SLAM. 

The basic idea in PEBO-SLAM is to translate state estimation into online \emph{constant} parameter identification. Though smoothing approaches rely on parameter estimation, we underline the radical difference between smoothing and the proposed method. The former relies on batch optimization, in which all historical poses are estimated in each step, known as full-information estimation; however, in the proposed method we design a dynamic extension to generate invariant foliations, and then obtain some ``invariance'', with only the current pose to be estimated on-line. The fact makes the proposed approach belong to the filtering category. 
 
The main contributions of the paper are threefold.
\begin{itemize}
    \item[\bf C1] Showing that the VI-SLAM problem can be formulated as online parameter estimation via a PEBO design on nonlinear manifolds, then obtaining a linear least squares problem without any approximation. 
    
    \item[\bf C2] Providing a simple VI-SLAM observer design, which enjoys almost global asymptotic stability, in contrast to locality in EKF-SLAM. Hence, it overcomes the issue of inconsistency and inaccuracy from bad initial guess, and enjoys low computational complexity with linear growth with respect to the number of features.
    
    \item[\bf C3] Relaxing significantly the PE or UCO condition required in some recent results, e.g., \cite{BJOetal,LOUetal,TANetal}, and only requiring an extremely weak assumption, i.e. interval excitation (IE). It means that PEBO-SLAM can still have satisfactory accuracy and consistency when the mobile robot moves along some \emph{non-uniformly} observable trajectories.   
\end{itemize}

An abridged version of was presented in the 60th IEEE Conference on Decision and Control \cite{YIetalcdc}. While the conference version primarily focused on the core idea, this extended version includes several additional contributions, including the proof of Lemma \ref{lem:1}, the convergence rate estimate in Theorem \ref{prop:mapping}, the proof of Theorem \ref{prop:pose-observer} and more simulation results.

\subsection{Notations}

Throughout the paper, the arguments and subscripts are omitted when clear from context. $\et$ represents exponentially decaying terms with proper dimensions. Given $A \in \rea^{n\times n}$ and $S\in \rea^{n\times n}_{\succeq 0}$, the Frobenius norm is defined as $\|A\| = \sqrt{\tr(A^\top A)}$, and $S^\hal$ and $\adj\{A\}$ represent the matrix square root and the adjugate matrix, respectively. The skew-symmetric projector $\pa(\cdot)$ is defined as $\pa(A) = \hal(A-A^\top)$. We use $|\cdot|$ to denote Euclidean norms of vectors or its induced matrix norm. Given two symmetric matrices $A$ and $B$ in $\rea^{n\times n}$, we write $A \preceq B$, if for all $\bfx\in\rea^n$, $\bfx^\top A \bfx \le \bfx^\top B \bfx$. We use $SO(3)$ to represent the special orthogonal group, and ${\mathfrak {so}}(3)$ is the associated Lie algebra as the set of skew-symmetric matrices satisfying $SO(3)=\{R\in \rea^{3\times3}|R^\top R = I_3, ~ \det(R) =1\}$. Given $\mathbf{a} \in \rea^3$, we define the operator $(\cdot)_\times$ as 
$$
\mathbf{a}_\times := \begmat{ 0 & - a_3 & a_2 \\  a_3 & 0 & -a_1 \\ -a_2 & a_1 & 0 } \in {\mathfrak {so}}(3) .
$$
We also consider the special Euclidean group denoted as $SE(3) = \{ \calt(R,\bfx)\in \rea^{4\times 4}|R\in SO(3),~ \bfx\in \rea^3\}$ with 
\begequ
\label{calt}
\calt(R,\bfx) = \begmat{R & \bfx \\ 0  & 1}.
\endequ
The Lie algebra of $SE(3)$ is defined as 
$$
{\mathfrak{se}}(3):= \left\{ A \in \rea^{4\times 4} \Bigg|A= \begmat{\Omegab_\times & \mathbf{v}\\ 0 & 0}, \Omegab_\times\in \mathfrak{so}(3), \bfv \in \rea^3 \right\}.
$$
For any $\bfx\in \rea^3$, $\bfa\times \bfx$ is the vector cross product, satisfying $\bfa_\times \bfx = \bfa \times \bfx$. For any $\bfx\in \rea^3/\{0\}$, its projector is defined as
$
\Pi_{\bfx} :=I_3 -{1\over|\bfx|^2}\bfx\bfx^\top,
$
which projects a given vector onto the subspace orthogonal of $\bfx$, since $\Pi_{\bfx} \bfx = 0$. We define a wedged mapping
\begequ\label{U}
\bfu^\wedge = \begmat{\Omegab_\times & \bfv \\ 0 & 0}
\endequ
for a vector $\bfu:=\col(\Omegab,\bfv)\in \rea^6$. 

%
\section{Preliminaries}
\label{sec2s}
%

To streamline the presentation, let us recall some preliminaries about the PEBO methodology in Euclidean space \cite{ORTetalscl}. We are particularly interested in estimation of unknown states via recursive algorithms. Consider the nonlinear model
\begin{equation}
\label{NLS}
\dot \bfx   =  f(\bfx ,\bfu)
, \quad 
\bfy   =  h (\bfx)
\end{equation}
with the systems state $\bfx \in \rea^{n_x}$, $\bfu\in \rea^{n_u}$ the input and $\bfy \in \rea^{n_y}$ the output. The target is to  asymptotically estimate unknown $\bfx$, assuming that both $\bfy$ and $\bfu$ are available. In the context of observer, the estimate $\hat\bfx  \in \rea^{n_x}$ is generated from  
\begin{equation}
\label{observer:general}
\begin{aligned}
 	\dot{\eta} & ~ = ~{F}( \eta,\bfy,\bfu)  \\
 	\hat \bfx & ~ = ~{N}(\eta,\bfy,\bfu) 
 \end{aligned}
\end{equation}
with two functions $F$ and $N$ to be designed, such that 
\begin{equation}
\label{convergence}
\lim_{t\to \infty} |\hat \bfx(t) - \bfx(t)| = 0.
\end{equation}
It can be designed in continuous time and then implemented via discretization, or be designed as $\eta_{k+1} = {F}(\eta_k,\bfy_k,\bfu_k)$ in discrete time directly.

The most popular estimators in the robotics community may refer to extended Kalman filter (EKF) and Luenberger-like observers \cite{BES}; see \cite{BERetalREV} for a recent review of nonlinear observers. In EKFs, the observer state $\eta$ consists of the estimate $\hat \bfx$ of unknown variables and the predicted covariance estimate $P$, and $F$ contains the prediction and update functions. In Luenberger-like observers, the state $\eta$ is usually selected as the estimate $\hat\bfx$ directly. If the system satisfies uniform complete observability along given robot trajectories, we may guarantee the convergence $|\hat \bfx - \bfx | \to 0$ locally with asymptotic stability. Note that in many cases the change of coordinates plays a very important roles in observer design \cite{KAZKRA,BES,YIetaltac22}.

In contrast to the design of an asymptotically stable error dynamics in EKF or Luenberger-like observers, a radically new method called PEBO was proposed in \cite{ORTetalscl}, which relies on generation of invariant foliations to translate the state estimation problem into one of estimation of constant, unknown parameters. Its first step is to find a smooth change of coordinate $\bfx\mapsto \bfz = \phi(\bfx) \in \rea^{n_z}$ ($n_z \ge n_x$), in which the time derivative is \emph{measurable}, i.e.
\begin{equation}
\label{dotz_pebo}
\dot \bfz = H(\bfy,\bfu)
\end{equation}
for some function $H$ with proper dimensions. It is equivalent to solve the partial differential equation (PDE)
\begin{equation}
\label{pde:pebo}
{\partial {\phib} \over \partial \bfx}(\bfx,\bfu){f}(\bfx,\bfu) = H({h}(\bfx),\bfu).
\end{equation}
Then, we establish, via introducing a pure integral action
\begin{equation}
\label{dot_xi}
\dot\xib = H(\bfy,\bfu),
\end{equation}
the relation
\begequ
\label{phi_xi_theta}
\phi(\bfx(t)) = \xib(t) + {\thetab}, \quad \forall t \ge 0
\endequ
with a \emph{constant} vector 
$
\thetab:= \phib(\bfx(0)) - \xib(0).
$ 
If $\thetab$ is known and $\phib$ is injective, we would get the true state as $ \bfx= \phil(\xib+\thetab, \bfy)$ with $\phil$ the left inverse of $\phib$. Hence, the remaining task is to identify the constant vector $\thetab$ online, denoted as $\hat \thetab$, then obtaining the estimate
$
\hat \bfx = \phil(\xib + \hat\thetab, \bfy).
$
To this end, we rely on the existence of the regression model\footnote{High-order derivatives of the output $y$ may also be involved to yield a regression model. For such a case, some linear time-invariant filter can be used to remove the difficulty of calculate dirty derivatives \cite{SASBOD}.}
\begequ
\label{output_regressor}
\bfy = {h}\big(\phil(\xib+\thetab,\bfy) \big)
\endequ
with the signals $\xib$ and $\bfy$ available. An online parameter estimator for the regression model \eqref{output_regressor} and the dynamic extension \eqref{dot_xi}, with the observer output $\hat \bfx = \phil (\xib + \hat\thetab , \bfy)$, is called a PEBO for the system \eqref{NLS}. In Section \ref{sec3}, we extend this idea from Euclidean space to the special Euclidean group, and then apply to the VI-SLAM problem.


\begin{remark}\rm The key idea in PEBO is to use the pure integral action \eqref{dot_xi} to generate an invariant foliation $\{(\bfx,\xib)\in \rea^{2n_x}| \phi(\bfx)= \xib+\thetab, \; \thetab \in \rea^{n_x}\}$ \cite{YIetal}, and then an online estimator is needed to identify the leaf of the foliation in which the systems state lives. Fig. \ref{fig:pebo} presents an intuitive illustration.
\begin{figure}[!htb]
\centering
\includegraphics[width=0.6\linewidth]{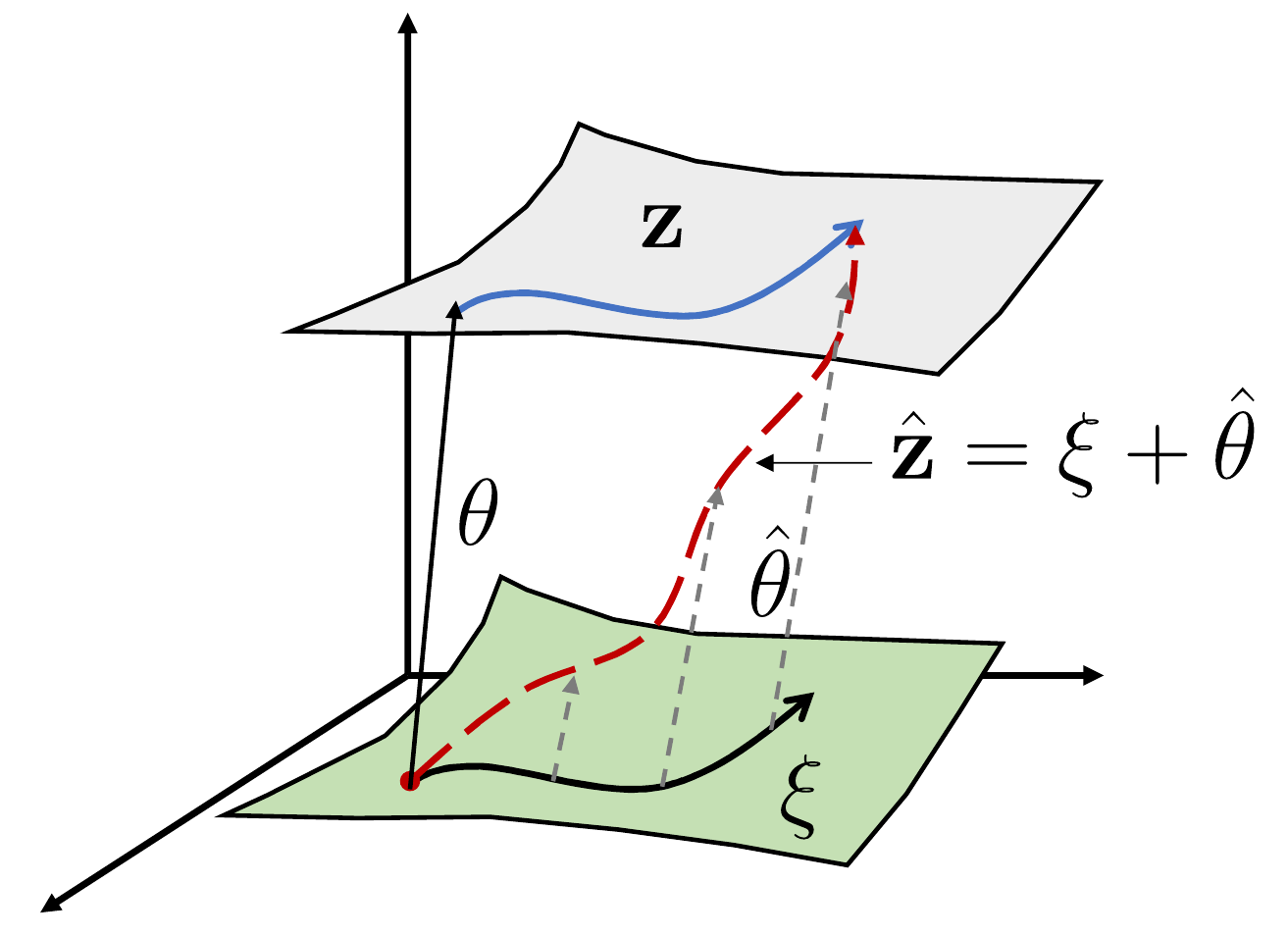}
\label{fig:pebo}
\caption{A geometric illustration of PEBOs}
\end{figure}
\end{remark}


\begin{remark}\rm
The PDE \eqref{pde:pebo} in PEBO is similar to the one in nonlinear contracting observers \cite{MAN,YIetaltac22}, i.e., ${\partial \phi \over \partial \bfx}(\bfx,\bfu) = H(\bfx,h(\bfx),\bfu)$. The difference relies on that in the latter $H(\cdot)$ is also a function of the transformation $\phi(\bfx)$, and we impose a contraction (a.k.a. incremental exponential stability) constraint in the new coordinate. This fact provides a possible way to extend the nonlinear contracting observer framework in \cite{YIetaltac22} only requiring incremental stability---but not asymptotic---in the transformed coordinate.
\end{remark}

%
\section{Problem Statement}
\label{sec2}
%
%

\subsection{Kinematic Model}

The kinematics of a rigid body robot with is given by
\begin{equation}
\label{kinematics}
	\dot \bfx   =  R \bfv , \quad 
	\dot R   =  R \Omegab_\times.
\end{equation} 
\begin{table}[!htb]
\begin{tcolorbox}[
colback=white!10,
coltitle=blue!20!black,  
]
\begin{center}
 {\bf \small Nomenclature}
\end{center}
\vspace{0.2cm}
  \renewcommand\arraystretch{1.4}
\small
\begin{tabular}{ll}
$\hat{(\cdot)}$ & Estimate of a variable or parameter \\
$\tilde{(\cdot)}$ & Estimation error\\
$I_{n}$ & {${n\times n}$ identity matrix} \\
$\small\{\cali,\calb,\calv\}$ & Frames of inertial, body and dynamic \\
& extension \\
$\bfx(t) \in \rea^3$ & Robot position at $t$ \\
$\bfv(t) \in \rea^3$ & Linear velocity in $\{\calb\}$ \\
$\Omegab(t) \in \rea^3$ & Rotational velocity in $\{\calb\}$\\
$R(t) \in SO(3)$ & Robot attitude matrix $\tensor[^I]{R}{_B}$ \\
$X(t) \in SE(3)$ & Rigid-body pose $X:=\calt(R,\bfx)$\\
${}^vX(t) \in SE(3)$ & Dynamic extended state ${}^vX  :=\calt(Q,\xib)$\\
${}^cX\in SE(3)$ & Constant rigid transformation ${}^v\mathbf{T}_I$ with \\
& the decomposition ${}^c X:= \calt({}^cQ ,{}^c\xib)$ 
\\
$Q(t)\in SO(3)$ & Attitude matrix in $\{\calv\}$\\
$\xib(t)\in \rea^3$ & Robot position in $\{\calv\}$\\
$\bfu(t) \in \rea^6$ & Velocity $\bfu=\col(\Omegab,\bfv)$ in $\{\calb\}$\\
${\ellb}_i, {}^v\ellb_i \in \rea^3$ & Position of the $i$-th landmark in $\{\cali\}$, $\{\calv\}$\\
$\bfy_i, {}^v\bfy_i \in \rea^3$ & Bearing of the $i$-th landmark in $\{\calb\}$, $\{\calv\}$ \\
$\caln \subset \mathbb{N}$ & Index set of landmarks $\{1,\ldots, n\}$
\end{tabular}
\end{tcolorbox}
\end{table}
All the definitions of symbols and the spaces where they live in can be found in the table of Nomenclature. We assume that there are $n$ feature points appearing in the field view of camera, the coordinates $\ellb_i$ of which in the inertial frame $\{\cali\}$ are constant, thus satisfying
\begin{equation}
\label{dyn:ldmk}
\dot \ellb_i = 0, \quad  i\in \caln :=\{1,\ldots, n\} \subset \mathbb{N}.
\end{equation}
The robot position $\bfx$ is the coordinate of the origin of the body-fixed frame $\{\calb\}$ relative to the origin of the inertial frame $\{\cali\}$. Here, we assume that the velocities $\bfv$ and $\Omega$ are bounded for $t\ge 0$, and guarantee the dynamics \eqref{kinematics} forward complete. The rigid body pose contains the attitude $R \in SO(3)$ and the position $\bfx \in \rea^3$, and with a slight abuse of notation we use $X \in SE(3)$ to represent the pose, which is given by $X:= \calt(R,\bfx)$ with the function $\calt$ defined in \eqref{calt}. Hence, we have its dynamics
\begin{equation}
\label{dot_X}
\dot X = X \bfu^\wedge
\end{equation}
with $\bfu=\col(\Omegab,\bfv)$ containing rotational and linear velocities. In the visual inertial SLAM problem only monocular cameras and IMUs are equipped on robotics, only providing the information of landmark bearings and the velocity $\bfu$. Then, the system output is given by
\begin{equation}
\label{output}
\bfy_i 
= R^\top {\ellb_i - \bfx \over |\ellb_i - \bfx|}.
\end{equation}
For convenience, we write all bearing measurements and landmark coordinates in vectors
$$
\begin{aligned}
\vbf{Y}  :=\begmat{\bfy_1 \\ \vdots \\ \bfy_n},
\quad 
\vbf{L}  :=\begmat{\ellb_1 \\ \vdots \\ \ellb_n}.
\end{aligned}
$$
Though the dynamics \eqref{dot_X} is linear, the output functions $ {h}_i$ are highly nonlinear and the state space is a nonlinear manifold, thus making \eqref{dot_X}-\eqref{output} a nonlinear system.

\subsection{Problem Statement}

In this paper, we are concerned with the VI-SLAM problem, i.e., estimating the pose $X(t)\in SE(3)$ and the landmark coordinates $\ellb_i\in \rea^3$ ($i\in \caln$) from the measurement \eqref{output} in real time, under the kinematic constraint \eqref{kinematics}, and the criterion is to make their estimates as close to the true values as possible.

As mentioned in Introduction, there are many merits to solve the above problem by employing the filter/observer approach, e.g., high computational efficiency and low data memory usage. Along this line, our task becomes the following deterministic observer design for VI-SLAM.

\ \\
{\bf  Problem 1.} ({\em Visual inertial SLAM observer}) Consider the dynamical model \eqref{kinematics} with the output $\vbf Y$ and the input $\bfu$. Design a continuous-time observer in the form
\begin{equation}
\label{obs:general}
 \begin{aligned}
 	\dot{ \eta} & ~ = ~ F(\eta, \vbf{Y},\bfu)\\
 	(\hat X,\hat{\vbf L}) & ~ = ~ N( \eta, \vbf{Y},\bfu) 
 \end{aligned}
\end{equation}
with the variables $\hat X \in SE(3)$ and $\hat{\vbf L} \in \rea^{3n}$, and the functions $F$ and $N$ to be designed, guaranteeing 
\begin{equation}
\label{convergence}
\lim_{t\to \infty}\Big[ \big\|\hat X(t) - X(t)\big\| + \big|\hat{\vbf L}(t) - {\vbf L}\big| \Big] = 0, 
\end{equation}
under some ``informative trajectories'' of the robot, and identify the convergence condition of the proposed observer.
\qed

For on-line estimation problems, it is widely recognized that excitation conditions can be used to characterize ``informative trajectories''. Here, let us recall some definitions as follows, including persistency of excitation (PE) and interval excitation (IE). Clearly, IE is strictly weaker than PE, since the former does not require uniformity in time.

\begin{definition}\label{def1}\rm
({\em PE, IE}) A bounded signal $\phi:\rea_+ \to \rea^n$ is
\begin{itemize}
  \item[-] $(T,\delta)$-PE, if
    \begin{equation}
    \label{def_pe}
     \int_{t}^{t+T} \phi(s)\phi^\top(s) ds \succeq \delta I_n, 
     \quad
     \forall t\ge 0
    \end{equation}
    for some $T>0, \delta >0$.

    \item[-] $(t_0,t_c,\delta)$-IE if there exist $t_0 \ge 0$ and $t_c \ge 0$ such that
    \begin{equation}
    \label{def_ie}
    \int_{t_0}^{t_0+t_c} \phi(s)\phi^\top(s) ds \succeq \delta I_n
    \end{equation}
    for some $\delta >0$. \qed
\end{itemize}
\end{definition}

An LTV system $\dot \bfx = A(t) \bfx$, $\bfy = C(t)\bfx$ is uniformly completely observable (UCO), if the matrix $[C(\cdot)\Phi_A(\cdot,t_0)]^\top$ is PE with $\Phi_A(\cdot,\cdot)$ the state transition matrix \cite{SASBOD}.

%
\section{PEBO-SLAM Framework}
\label{sec3}
%
\subsection{A Parameter Estimation Approach to VI-SLAM}

In this section, we propose the PEBO-SLAM framework to achieve \eqref{convergence}.

As shown above, one of the key steps in PEBO is to generate invariant foliations among the system states and dynamic extensions. In our case, the systems state evolves on the nonlinear manifold $SE(3) \times \rea^{3n}$. In the SLAM context, we construct a dynamic extension
\begin{equation}
\label{dyn_ext1}
{}^{v}\dot {X} = {}^vX \bfu^\wedge
\end{equation}
with the decomposition
\begin{equation}
\label{Xv}
{}^vX:=\calt(Q,\xib)\in SE(3),
\end{equation}
which resembles the pure integral action \eqref{dot_xi} in Euclidean space. We refer the pose ${}^vX$ living in the dynamic extension frame $\{\calv\}$. 

We have the algebraic relation in Lemma \ref{lem:1} below, which may be viewed as an extension of the key identity \eqref{phi_xi_theta} in PEBO from Euclidean space to matrix Lie groups. A similar idea was employed in our previous work \cite{yi2023attitude}.
\begin{lemma}
\label{lem:1} \rm
Consider the dynamics \eqref{kinematics}. The dynamic extension \eqref{dyn_ext1} is forward complete, and there exists a constant matrix ${}^c X = \calt({}^c Q,{}^c\xib) \in SE(3)$ satisfying  
\begin{equation}
\label{id:1}
{}^vX(t) ={}^c  X X(t), \quad \forall t\ge 0,
\end{equation}
with ${}^cQ = Q(0)R(0)^{-1}$ and ${}^c \xib = \xib(0) - {}^c Q \bfx(0)$.
\end{lemma}
\begin{proof}
Define an error variable 
$
E(X,{}^vX):={}^vX  X^\top,
$
which lives on $SE(3)$. Its time derivative is given by
$$
\dot E = {}^vX \bfu^\wedge X^{-1} -{}^v X X^{-1}\dot X X^{-1} = 0.
$$
Thus, there exists a constant matrix ${}^c X \in SE(3)$ satisfying 
$$
E(X(t),{}^v X(t))={}^c X, \quad \forall t\ge 0,
$$
which may be written equivalently as \eqref{id:1}.\footnote{This relationship is called in \cite{LAGetal} that the systems \eqref{dot_X} and \eqref{dyn_ext1} are $E$-synchronous.} By re-arranging terms in ${}^v X(0)={}^c X X(0)$, it is straightforward to get 
$$
Q_c :=Q(0)R(0)^{-1}
, \quad
{}^c\xib := \xib(0) - {}^cQ x(0).
$$

From the forward completeness assumption of the dynamics \eqref{kinematics}, as well as invoking the identity \eqref{id:1}, we conclude that the dynamic extension \eqref{dyn_ext1} is also forward complete.
\end{proof}

\vspace{0.2cm}

The above lemma shows that the open-loop dynamic extension \eqref{dyn_ext1} and the kinematics \eqref{kinematics} admit an affine relationship \eqref{id:1}---more precisely, there is a constant rigid transformation ${}^v\mathbf{T}_I$ between the frames $\{\cali\}$ and $\{\calv\}$---by means of which we reformulate the state estimation of $X(t)$ into the problem of online \emph{constant} parameter identification of ${}^cX \in SE(3)$.

In the following, we show the unknown constant matrix ${}^cX=\calt({}^cQ,{}^c\xib)$ verifies a nonlinear regression model.

\begin{lemma}
\label{lem:2}\rm
The functions of unknown variables $({}^cX,\vbf{L})$, defined by
\begin{equation}
\label{const:ldmk}
{}^v \ellb_i := {}^c\xib + {}^cQ\ellb_i,
\end{equation}
are \emph{constant}, and verify the algebraic equation
\begin{equation}
\label{bearing:virtual}
 \bfy_i = Q^\top {{}^v \ellb_i - \xib \over |{}^v \ellb_i - \xib|},
\end{equation}
with $Q,\xib$ given in \eqref{Xv}. 
\end{lemma}


\begin{proof}
According to \eqref{id:1} and \eqref{const:ldmk}, we have
$$
\begin{aligned}
{}^v \ellb_i & ~ = ~ {}^c\xib_c + {}^c Q \ellb_i
\\
& ~ = ~ [{}^c\xib + {}^cQ\bfx] + [{^c}Q R] R^\top [\ellb_i -\bfx]  
\\
& ~ = ~ \xib + QR^\top (\ellb_i - \bfx),
\end{aligned}
$$ 
and thus, from some basic geometric relation, ${}^v\ellb_i$ can be viewed as the feature point coordinates in the dynamic extension frame $\{\calv\}$. Correspondingly, the bearings in $\{\calv\}$ can be defined as
$$
 {}^v\bfy_i := Q^\top {{}^v \ellb_i - \xib \over |{}^v \ellb_i- \xib|}.
$$
Some simple calculations show that
$$
\begin{aligned} 
	{}^v\bfy_i & ~ = ~ Q^\top{QR^\top(\ellb_i - \bfx) \over |QR^\top(\ellb_i - \bfx)|}
	\\
	& ~  =  ~ R^\top {\ellb_i -\bfx \over |\ellb_i - \bfx|} 
	\\
	& ~ = ~ \bfy_i,
\end{aligned}
$$
thus yielding ${}^v\bfy_i \equiv \bfy_i$.\footnote{An intuitive interpretation to ${}^v\bfy_i \equiv \bfy_i$ is that the transformation \eqref{id:1} of the ambient from $\{\cali\}$ to $\{\calv\}$ does not change the \emph{relative} transformation from a robot to a landmark.} Hence, we have verified \eqref{bearing:virtual}.
\end{proof}

\vspace{0.2cm}

\begin{remark}
\rm 
Eqs. \eqref{bearing:virtual}-\eqref{const:ldmk} provides a nonlinear regressor w.r.t. the unknown constant matrix ${}^cX$, with $\bfy_i$, $Q$ and $\xib$ \emph{available}. Despite being quite simple, the above lemmata show estimation of time-varying states may be translated into online parameter estimation of ${}^c X$ and ${}^v\ellb_i$ successively. To the best of the authors' knowledge, such a fact has not been explored in SLAM filters design, which, however, provides possibility to relax significantly the sufficient excitation conditions in existing approaches, and to get global convergence to overcome inconsistency and inaccuracy from bad initial guesses. 
\end{remark}

In terms of the above lemmata, we reformulate SLAM into the following consistent online estimation problem.

\ \\
{\bf Problem 2} ({\em PEBO-SLAM}) Consider the dynamic extension \eqref{dyn_ext1}, with available signals ${}^vX, \bfu$ and $\bf Y$. Design a globally convergent online identifier for ${}^v\ellb_i$ and ${}^cX$ recursively. \qed

If the above problem is solvable with the estimates ${}^v\hat \ellb_i$ and ${}^c \hat X$, then 
$$
\hat X = {}^c \hat X^\top {}^vX,
\qquad
\begmat{\hat \ellb_i \\ 1}={}^c \hat X^\top \begmat{{}^v\hat \ellb_i \\ 1}
$$
provide a feasible solution to {\bf Problem 1}. It is well known that the online estimation of constant parameters is much easier than the one of time-varying states. This is one of the motivations of the paper. We call the parameter identifier (to be designed) with the dynamic extension \eqref{dyn_ext1} as PEBO-SLAM.

 Throughout the paper, we need the following assumptions.
\begin{assumption}
\label{ass:ie}\rm 
The robot trajectory $X$ guarantees that 1) $\Pi_{Q\bfy_i}$ is IE for all the landmarks to be mapped; and 2) the origin of $\{\calb\}$ never coincides with any feature points.
\end{assumption}


Let us define the vectors 
$$
\bfr_i := \ellb_{i+1} - \ellb_{i}, ~ \forall i \in \caln\backslash\{n\},
$$
and see Fig. \ref{fig:coordinates} for illustration. 

\begin{assumption}
\label{ass:3ldmk}\rm
There exists $n_\ell \ge 3$ such that  
$
\bfr_i \times \bfr_j \neq 0,~ \mbox{for~} i\ne j, ~ i,j\in \{1, \ldots,n_\ell\}.
$
\end{assumption}

\begin{figure*}[!htp]
    \centering
    { \includegraphics[width=0.6\textwidth]{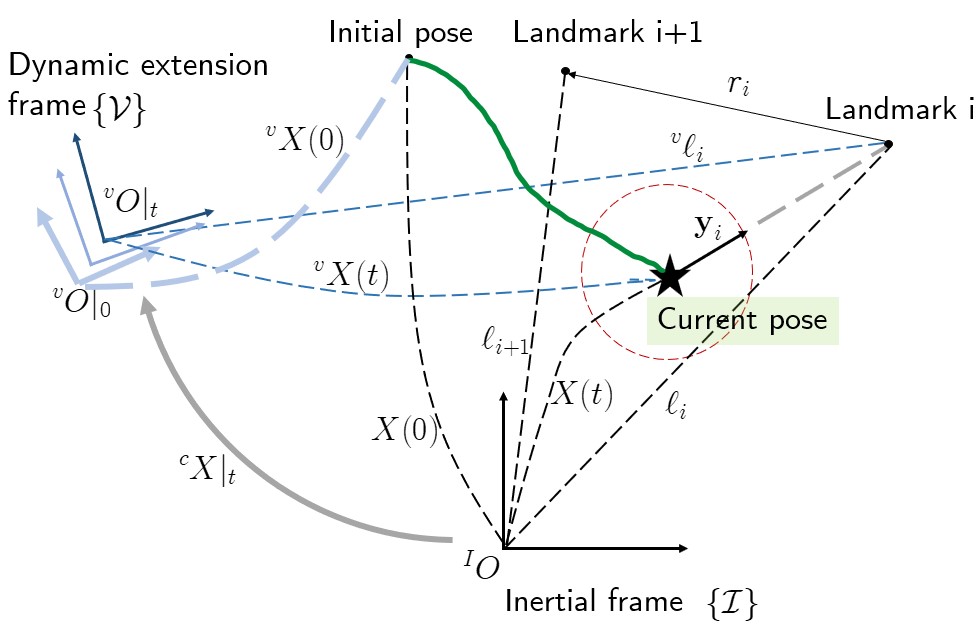} }
    \caption{Illustration of the frames in the proposed design (The dynamic extension frame is a fixed one with ideal measurements; in the above figure, it is a slowly moving frame due to the drifts stemmed from odometry errors.)}
    \label{fig:coordinates}
\end{figure*}

\subsection{A Linear Least Squares Solution}
\label{sec:32}

In this subsection, we show that the nonlinear regressors \eqref{bearing:virtual} can be equivalently written as linear regressors w.r.t. ${}^v \ellb_i$.

When the $i$-th feature point is visible in the camera view, invoking \eqref{bearing:virtual} yields for $\forall t\ge 0$ 
$$
Q \bfy_i   = {{}^v \ellb_i - \xib \over |{}^v \ellb_i - \xib |},
$$
and after straightforward calculations we obtain
$$
Q \bfy_i [Q \bfy_i]^\top ({}^v \ellb_i - \xib)
  = {}^v \ellb_i - \xib.
$$
Noting that $\bfy$ is a unit vector (i.e. $|\bfy|=1$), we thus obtain
$$
\Pi_{Q \bfy_i} \cdot [{}^v \ellb_i - \xib] =0, 
$$
which can be expressed as the linear regression models
\begin{equation}
\label{lre1}
\bfq_i(t) = \phi_i^\top(t) \cdot {}^v \ellb_i,
\end{equation}
with the measurable regression matrix $\Pi_{Q(t)\bfy_i(t)}$ and the regressor outputs $q_i$, which are defined as $\phi_i  := \Pi_{Q\bfy_i}^\top$ and $\bfq_i  := \Pi_{Q\bfy_i} \cdot \xib$, respectively. 

If the feature point is unavailable in the view during a bounded interval $\cali \subset [0,\infty)$, we may define $\phi_i:= 0_{3\times 3}$ and $\bfq_i=0_3$. Then, it is easy to verify that the linear regressors \eqref{lre1} still hold true, though not providing any information.

Hence, we unify all landmarks in the linear regression models \eqref{lre1} with the signals
\begequ
\label{phiiqi}
\begin{aligned}
\phi_i & ~ :=~ \left\{ 
\begin{aligned}
& \Pi_{Q(t)\bfy_i(t)}^\top & ~~~~~~ &\mbox{visible}
\\
& 0_{3\times 3} & &\mbox{invisible}
\end{aligned}
\right.
\\
\bfq_i & ~:=~ \left\{ 
\begin{aligned}
&\Pi_{Q(t)\bfy_i(t)}  \xib(t) &~ &\mbox{visible}
\\
& 0_3& & \mbox{invisible} 
\end{aligned}
\right.
\end{aligned}
\endequ

For the $i$-th landmark, the estimation of ${}^v\ellb_i$ can be solved via the linear least squares problem
\begequ
\label{LLS1}
 {}^v\hat\ellb_i ~=~ \underset{{}^v\ellb_i \in \rea^3}{\arg\min} \sum_k \Big|  \Pi_{Q(k)\bfy_i(k)}\cdot (\xib(k)-{}^v\ellb_i) \Big|^2
\endequ
for a few moments in the visible interval $k \in [t_a,t_b]$.

Note the above least squares problem is solvable in at least two steps, since the non full rankness of $\Pi_{Q\bfy_i}$ at any single moment. Indeed, its rank equals to two all the time when being visible. Note that the above least squares problem is presented in the deterministic sense. When considering the stochastic noise with Gaussian distribution of the bearing output \eqref{output}, the distribution cannot be preserved in the linear regression model \eqref{lre1} after nonlinear transformation.

\begin{remark}
In the above analysis, we use the projector to obtain the linear regression models \eqref{lre1} from bearings in the dynamic extension frame. This technical trick was widely used in the literature, e.g., bearing-only formation, range estimation \cite{HAMSAM}, and navigation \cite{WANetal2021}. An alternative to generate a linear regressor may be found in \cite{LOUetal} by augmenting both the position ${}^v \ellb_i$ and its range $| {}^v \ellb_i |$ in the systems state. Recently, a new linear parameterization to the feature range in the body-fixed frame $\{\calb\}$ is proposed in \cite{YIetalAUT} using linear stable filters.
\end{remark}

\begin{remark}\rm 
\label{remark:proc}
A sketch of key steps in the proposed PEBO-SLAM framework is given as follows, based on which we present a feasible observer solution in the next section.
\begin{itemize} 
\item[\bf S1] Implement the dynamic extension
$$
{}^v\dot X = {}^v X \bfu^\wedge, \quad \forall t\ge 0,
$$
with $\bfu$ defined in \eqref{U} the velocities from IMUs, and the pose in the frame of dynamic extension ${}^v X = \calt(Q,\xib)$.

\item[\bf S2] For each visible landmark, solve the least square problem in \eqref{LLS1} for a few moments in the visible interval $k \in [t_a,t_b]$.

\item[\bf S3] Fuse information of ${}^v X$ and $ {}^v \hat\ellb_i$ in order to get robust estimates $\hat X$ and $\hat \ellb_i$ in the inertial frame, in particular for the case with imperfect odometry.
\end{itemize}

\end{remark}
%
\section{A Feasible Observer Design}
\label{sec4}
%

After formulating VI-SLAM as a linear least squares problem, there are many approaches to solve it. In practice, measurement noise should be taken into account, and thus the dynamic extension \eqref{dyn_ext1} may have an accumulation of errors, thus making ${}^cX$ a \emph{slowly time-varying parameter} rather than a real constant. To alleviate such a robustness issue, we focus on designing consistent on-line observers.

In this section, we propose a feasible observer design to solve {\bf Problem 2}. The scheme consists of two cascaded observers, i.e., a mapping observer in $\{\calv\}$ in Theorem \ref{prop:mapping}---which estimates the landmarks ${}^v \ellb_i$ in the dynamic extension frame---and a localization and mapping observer, summarized in Theorem \ref{prop:pose-observer}, to express the estimates of pose and landmarks in the inertial frame $\{\cali\}$. See Fig. \ref{fig:schematic} for the overall schematic block diagram of the proposed design.

\begin{figure}[h]
    \centering
    \includegraphics[width=0.5\textwidth]{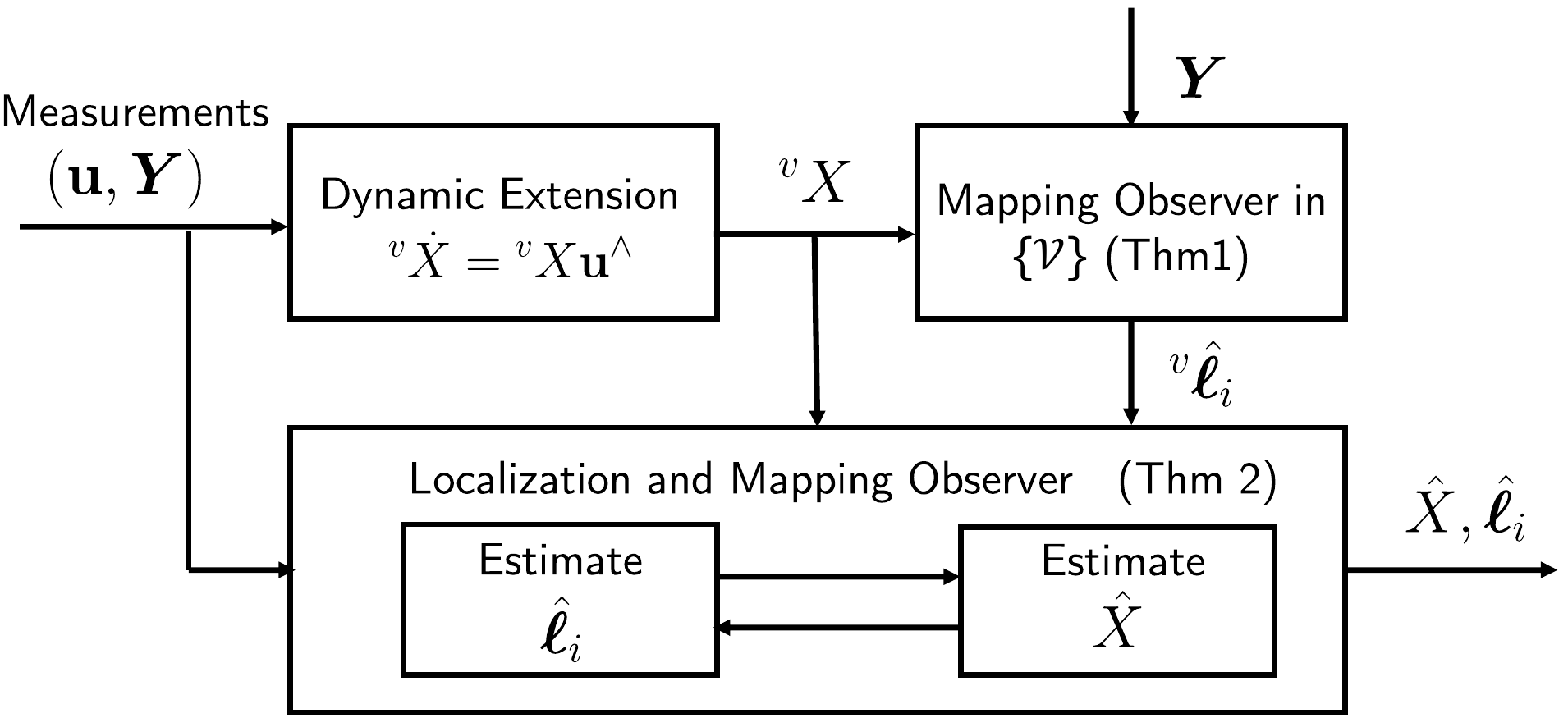}
    \caption{Block diagram of the proposed visual inertial SLAM observer}
    \label{fig:schematic}
\end{figure}

\subsection{Mapping Observer in the Dynamic Extension Frame}
\label{sec:41}

In this subsection, we consider on-line observer design to estimate ${}^v\ellb_i$, which corresponds to the step {\bf S2} in Remark \ref{remark:proc}.

It is well known that standard gradient descent or least square algorithms require some PE conditions to guarantee consistent convergence \cite{SASBOD}. Intuitively, it requires that landmarks appear in camera view persistently, and the robot keeps moving w.r.t. landmarks, which is rarely satisfied in practice. In the sequel, we show that in terms of the LREs \eqref{lre1} the restrictive excitation conditions can be relaxed significantly. The main technical route is using the DREM technique \cite{ARAetaltac,ORTetalnew} to obtain decoupled regressors, and implicitly combining with ``integral type'' historical information. We have the following.

\vspace{0.2cm}

{\em Mapping observer in the dynamic extension frame:} 
\begequ
\label{operator}
\dot\eta_\Sigma= F_\Sigma (\eta_\Sigma, \bfq_i, \phi_i),
\quad
{}^v \hat\ellb_i = N_\Sigma(\eta_\Sigma),
\endequ
in which the exogenous signals $\phi_i, \bfq_i$ are introduced in \eqref{phiiqi} with the observer internal state $\eta_\Sigma: = ({\bfq_i^3, \Phi_i, \chib_i, \omega_i, {}^v \hat \ellb_i})$, and the functions $F_\Sigma$ and $N_\Sigma$ are given by the state-space model
\begin{equation}
\label{ldmk-obs}
\Sigma:
\left\{~
\begin{aligned}
{\dot \bfq}_i^e & ~=~- \alpha_i \bfq_i^e + \alpha_i \phi_i \bfq_i, ~\bfq_i^e(0)=0
	\\
	\dot{\Phi}_i & ~=~ - \alpha_i\Phi_i+ \alpha_i \phi_i, ~ \Phi_i(0) = 0  
	\\
	\dot{\chib}_i & ~=~  \Delta_i ( \vbf{Y}_i - \Delta_i \chib_i )
	\\
	\dot{\omega}_i & ~=~ - \Delta_i^2\omega_i, ~\omega_i(0)=1
	\\
	{}^v \dot{\hat \ellb}_i & ~=~  \gamma_i \Delta_i^e  ( \vbf{Y}_i^e - 
	\Delta_i^e \cdot 	{}^v\hat{\ellb}_i)
\end{aligned}
\right.
\end{equation}
with 
\begin{itemize}
    \item[-] the functions:  
\begequ
\label{delta_ie}
\begin{aligned}
\Delta_i  & ~:=~ \det\{\Phi_i\}
\\
\vbf{Y}_{i} & ~:= ~\adj\{\Phi_i\} \bfq_i^e
\\
\Delta_i^e & ~:=~ \Delta_i + \ki^i(1-\omega_i)
\\
{\vbf Y}_i^e & ~:=~ \vbf{Y}_i + \ki^i\big[\chib_i - \omega_i\chib_{i}(0) \big].
\end{aligned}
\endequ
\item[-] the gains: 
$
\alpha_i>0, ~\gamma_i >0, ~\ki^i>0.
$
\end{itemize}
\qed

The above observer enjoys the following provable convergence guarantees.

\begin{theorem} \rm  
\label{prop:mapping} ({\em Mapping in the dynamic extension frame})
The mapping observer \eqref{ldmk-obs}-\eqref{delta_ie} guarantees
\begin{itemize}
    \item[1)] ({\em Internal stability}) All the internal states are bounded.
    \item[2)] ({\em Element-wise monotonicity}) For all $ t_a \ge t_b \ge 0$,
    $$
    |{}^v\hat \ell_{i,j} (t_a) - {}^v\ell_{i,j} | \le |{}^v\hat \ell_{i,j}(t_b) - {}^v \ell_{i,j}|.
    $$
    \item[3)] ({\em GES under IE}) Assuming that $\Pi_{Q\bfy_i}$ is $(t_0,t_c,\delta)$-IE, we have
    \begequ
    \label{convergence:tilde_ell}
    \lim_{t\to\infty} | {}^v \hat \ellb_i (t) - {}^v\ellb_i| =0 \quad \mbox{(exp.)}
    \endequ
    globally. 
    \item[4)] ({\em Convergence rate}) Under the above IE condition, the exponential convergence rate is not less than
\begequ
\gamma_\star :=  \gamma_i (k_{\tt I}^i)^2  \left( 1 - \exp \left(- {\delta_0^2 \over 6\alpha_i e} \right) \right)^2
\label{gamma_star},
\endequ
after some moment with $\delta_0:=(\alpha_i \delta e^{-\alpha_i t_c})^3$, i.e.
$$
\begin{aligned}
\left|{}^v \tilde{\ell}_{i,j}(t) \right| \le  e^{-\gamma_\star (t-\tau_\star)} 
\left|
{}^v \tilde{\ell}_{i,j} 
\left(
t_0+t_c+{1\over 6\alpha_i}
\right)
\right|
,\qquad\qquad \\ t\ge t_0+t_c+ {1\over 6\alpha_i}.
\end{aligned}
$$
\end{itemize}
\end{theorem}

\vspace{0.1cm}

\begin{proof}
The first step is to generate a \emph{square} regression matrix, which may be full rank at some moment, by introducing linear bounded-input-bounded-output (BIBO) operators.%
\footnote{
To be precise, though in our case the regression matrix $\phi_i$ defined in \eqref{phiiqi} is square, it is not full rank, i.e.
$
\det\{\Pi_{Q(t)\bfy_i(t)}\} = 0,\; \forall t\ge 0.
$
Therefore, it implies the necessity of additional dynamic extensions to continue the design. 
}
It is easy to get
$$
\begin{aligned}
\dot{\aoverbrace[L1R]{\bfq_i^e - \Phi_i{}^v \ellb_i}}
& ~=~
-\alpha_i (\bfq_i^e - \Phi_i{}^v \ellb_i) +\alpha_i( \phi_i \bfq_i - \phi_i {}^v\ellb_i )
\\
& ~ = ~
-\alpha_i (\bfq_i^e - \Phi_i{}^v \ellb_i) +\alpha_i( \phi_i \phi_i^\top {}^v\ellb_i  - \phi_i {}^v\ellb_i )
\\
& ~=~
-\alpha_i \big( \bfq_i^e - \Phi_i {}^v \ellb_i\big),
\end{aligned}
$$
where in the second equation we have used \eqref{lre1}, and the last we have used the facts 
$$
\phi_i^\top = \phi_i, \quad \phi_i\phi_i^\top = \phi_i
$$ in terms of the definition $\bfq_i  := \Pi_{Q\bfy_i} \cdot \xib$. Hence, it yields a group of new extended LREs%
\begin{equation}
\label{k-elre2}
 \bfq_i^e = \Phi_i {}^v \ellb_i + \et,
\end{equation}
which is known as the Kreisselmeier's extended LRE \cite{KRE}, recalling that $\et$ is an exponentially decaying term. Then, we mix the regressors \eqref{k-elre2} to get three decoupled, \emph{scalar} regressors for each $i \in \caln$, that is pre-multipying the adjugate matrix $\adj\{\Phi_i\}$ to the both sides, thus obtaining
\begin{equation}
\label{decp-regr}
   Y_{i,j} = \Delta_i {}^v\ell_{i,j} + \et,  \quad  j\in\{1,2,3\}
\end{equation}
with the definitions
\begin{equation}
\label{yij}
\begin{aligned}
\vbf{Y}_{i}   := \adj\{\Phi_i\} \bfq_i^e, \quad
\Delta_i  := \det\{\Phi_i\}
\end{aligned}
\end{equation}
and 
$
\vbf{Y}_{i}:=
\col(Y_{i,1}, Y_{i,2}, Y_{i,3}).
$

The second step of the proof is to show that from \eqref{decp-regr} we may obtain  $3\times n$ new \emph{scalar} linear regressors. For convinience, we define $\chib_i(0) = \chib_{i,0}$. The time derivative of $(\chib_i - {}^v \ellb_i)$ is an LTV system, which is given by
$$
{d\over dt}[{ \chib_i -{}^v \ellb_i }]
=
- \Delta_i^2(t) [\chi_i -{}^v \ellb_i]
$$
with $\Delta_i$ a scalar bounded signal. Its solution is
$$
\begin{aligned}
\chib_i(t) -{}^v \ellb_i 
&
~=~ \exp\bigg( - \int_{0}^t\Delta_i(s)^2 ds  \bigg)  [\chib_i(0) -{}^v \ellb_i]
\\
&
~=~ \omega_i(t) (\chib_i(0) -{}^v \ellb_i) 
\end{aligned}
$$
by noting that the solution to the fourth equation in \eqref{ldmk-obs} is given by
$$
\omega_i(t) = \exp\left(- \int_{0}^t\Delta_i(s)^2 ds \right).
$$
It yields
\begin{equation}
\label{int_regr}
  \chib_i(t) - \omega_i(t) \chib_{i,0} = [1-\omega_i(t)] {}^v \ellb_i,
\end{equation}
in which we underscore that $\chib_i,\omega_i$ and $\chib_{i,0}$ ($t\ge 0$) are all available in observer design. In this way, we get the scalar LREs \eqref{int_regr} involving the ``integral'' information of $\Delta_i$. Combining \eqref{int_regr} and \eqref{decp-regr}, we obtain new linear regressors\footnote{The role of the parameters $\ki^i$ will be discussed in the sequel.}
\begin{equation}
\label{lre_mixing}
\vbf{Y}_i + \ki^i ( \chi_i - \omega_i\chib_{i,0} ) + \et
=
\Delta_i^e \cdot {}^v \ellb_i,
\end{equation}
in which ${}^v \ellb_i$ are the unknown constant parameters. Correspondingly, by defining the estimation error 
$$
{}^v \tilde \ellb_i := {}^v \hat\ellb_i -{}^v \ellb_i
$$
we have\footnote{We omit the exponentially decaying term $\et$ due to the specific initialization of $q_i^e$ and $\Phi_i$.}
\begequ
\label{dot_ziv}
{}^v\dot{\tilde \ellb}_i= - \gamma_i (\Delta_i^e)^2\cdot {}^v\tilde{\ellb}_i .
\endequ

Now we are in position to verify the first claim. Choosing the Lyapunov function 
$$
V(\chib_i,\omega_i,{}^v\tilde{\ellb}_i) = \hal\sum_{i=1}^n \left(
|{}^v\tilde{\ellb}_i|^2 + |\chib_i - {}^v \ellb_i|^2 + |\omega_i|^2
\right),
$$
it yields
$$
\begin{aligned}
\dot V & = -
\sum_{i=1}^n  \left[\gamma_i (\Delta_i^e)^2 |{}^v\tilde \ellb_i|^2 + \Delta_i^2 |\chib_i - {}^v\ellb_i|^2 + \Delta_i^2 \omega_i^2
\right]
 \le 0.
\end{aligned}
$$
Thus, the system is internally stable.

The second claim is verified invoking $\Delta_i^e$ are scalar signals. Recalling the dynamics of ${}^v \tilde \ellb_i$ in \eqref{dot_ziv}, the last claim is equivalent to verify $\Delta_i^e \in \mbox{PE}$. To simplify the presentation, we neglect the index $i$ of the IE condition from the $i$-th feature points. From the IE assumption of $\Pi_{Q\bfy_i}$, there exist $t_0,t_c,\delta \in \rea_{+}$ such that
$$
\int_{t_0}^{t_0+t_c} \Pi_{Q(s)\bfy_i(s)}\Pi_{Q(s)\bfy_i(s)}^\top ds
\ge \delta I_3.
$$
Noting the second equation in \eqref{ldmk-obs}, we can verify
$$
\begin{aligned}
 \Phi_i(t_0+t_c)  
~ \succeq  ~ &  \alpha_i\int_{t_0}^{t_0+t_c} e^{-\alpha_i(t_0+t_c-s)} \phi(s) ds
\\
~ \succeq ~ &  \alpha_i e^{-\alpha_i t_c} \int_{t_0}^{t_0+t_c}  \phi(s) ds
\\
~ = ~ & \alpha_i e^{-\alpha_i t_c} \int_{t_0}^{t_0+t_c}  \Pi_{Q(s)\bfy_i(s)} \Pi_{Q(s)\bfy_i(s)}^\top ds
\\
~ \succeq~ & \alpha_i \delta e^{-\alpha_i t_c}  I_3.
\end{aligned}
$$
It implies 
\begequ
\label{delta0}
\Delta(t_0+t_c) \ge \delta_0 := (\alpha_i \delta e^{-\alpha_i t_c})^3.
\endequ
Hence, from the continuity of differential equations, in a small neighborhood of the moment $t_0+t_c$, the signal $\Delta(t)$ is lower bounded from zero. Namely, there exist small $\tau>0$ and $\varepsilon >0$ such that 
\begequ\label{verify_PE}
\begin{aligned}\footnotesize 
 \int_{t_c'-\tau}^{t_c'} \Delta_i(s) ds  \ge \sqrt \varepsilon \quad
\\
~ \implies ~~~ & 
\int_{0}^{t} \Delta_i^2(s) ds \ge{1 \over \tau}\varepsilon, ~\forall t\ge t_c'
\\
 ~ \implies ~~~ &  1- \omega_i \ge 1-e^{ - {1\over \tau}\varepsilon}, ~\forall  t\ge t_c'
\\
~ \implies ~~~ &  \Delta_i^e= \Delta_i + \ki^i(1-\omega_i) \in \mbox{PE}
\end{aligned}
\endequ
with 
$
t_c':= t_0+t_c,
$
in which we have used the non-negativeness property $\Delta_i(t) \ge 0$, $\forall t\ge 0$ and the Cauchy-Schwarz inequality for integrals. For the LTV system \eqref{dot_ziv}, it is well known that the global exponential stability (GES) is equivalent to $\Delta_i^e$ being PE, which has been verified in \eqref{verify_PE}. Therefore, the error dynamics admits the exponential convergence \eqref{convergence:tilde_ell} under the IE condition of $\phi_i$. 

At the end, we give an estimate of the convergence speed of the proposed mapping observer. In terms of the above analysis, we have $\Delta_i(t_c) \ge \delta_0$ with the parameter $\delta_0$ defined in \eqref{delta0}. It is clear that if the matrix $\Pi_{Q{\bf y}_i}$ is $(t_{0},t_{c},\delta)$-IE, then it is also $(t_{0} ,t_{c} + \tau,\delta)$-IE for any $\tau\ge0$ according to the definition of IE. In the same way, we obtain
$$
\begin{aligned}
\Delta_i(t_0 + t_c + \tau) &~ \ge~ \left[\delta e^{-\alpha_i (t_0+t_c+ \tau)} \right]^3 
\\
& 
~=~ \delta_0 e^{-3\alpha_i \tau} 
, \quad \forall \tau\ge 0,
\end{aligned}
$$
and
$$
\Delta_i(t) \ge \delta_0 e^{-3\alpha_i \tau}, \quad \forall t\in [t_0+t_c, t_0+t_c+\tau].
$$
It is underlined that the parameter $\delta_0$ is independent of $\tau$. From the dynamics of $\omega$, it yields for $t\ge t_0+t_c+\tau$
$$
\begin{aligned}
\omega_i(t) 
& ~=~ \exp \bigg( - \int_0^t \Delta_i(s)^2 ds \bigg)
\\
& ~\le~ 
\exp \bigg( - \int_{t_0+t_c}^{t_0+t_c+\tau} \Delta_i(s)^2 ds \bigg)
\\
& ~\le~
\exp \bigg( -  \delta_0^2 \tau e^{-6\alpha_i \tau} \bigg).
\end{aligned}
$$
Then, invoking the definition of $\Delta_i^e$ we have for $t\ge t_0+t_c+\tau$
$$
\begin{aligned}
\Delta_i^e(t)  & ~ \ge ~ k_{\tt I}^i (1-\omega_i(t))
\\
 & ~ \ge ~ k_{\tt I}^i \left[ 1- \exp \bigg( -  \delta_0^2 \tau e^{-6\alpha_i \tau} \bigg) \right],
\end{aligned}
$$
where in the first inequality we have used the positive semi-definiteness of $\Phi(t)$. By investigating the stationary point of the function 
$$
f(\tau) = 1- \exp\left(-\delta_0^2 {\tau\over e^{6\alpha_i \tau}}\right),
$$
we have $\nabla f({1\over 6\alpha_i}) =0$ thus the function getting its unique minimum at $\tau_\star = {1\over 6\alpha_i}$. Hence,
$$
\Delta_i^e(t) \ge k_{\tt I}^i  \left( 1 - \exp \left(- {\delta_0^2 \over 6\alpha_i e} \right) \right), 
\quad
t \ge t_0+t_c+ \tau_\star.
$$
It is easy to get
$$
|{}^v \tilde{\ell}_{i,j}(t)| \le  e^{-\gamma_\star (t-\tau_\star)} |{}^v \tilde{\ell}_{i,j}(t_0+t_c+\tau_\star)|, \quad t\ge t_0+t_c+ \tau_\star,
$$
with $\gamma_\star$ defined in \eqref{gamma_star}. It completes the proof.
\end{proof}

\vspace{0.1cm}

In the above, we present a GES landmark observer using the extremely weak IE condition, and give an estimate of its convergence speed. After obtaining $({}^v X, {}^v\ellb_1, \ldots, {}^v\ellb_n)$, we indeed achieve the visual inertial odometry. Compared with VI-SLAM, the odometry from ${}^v X$ is highly sensitive to unavoidable measurement noise, and we will propose an observer in Section \ref{sec42} to complete the task to estimate ${}^c X$, as well as robustifying the design.

\begin{remark}\rm
Some remarks are made about Theorem \ref{prop:mapping}.

\begin{enumerate}
\item[(a)]  It is interesting to note that the proposed observer for \emph{each landmark} is decoupled, thus having computation complexity $\calo(n)$ with $n$ landmarks. This should be compared with $\calo(n^2)$ for standard EKF-SLAM algorithms \cite{TANetal}.

\item[(b)] All functions in the observer are $\calc^\infty$. This is unlike the approaches via homogeneous domain (with fractional power), where non-smooth or discontinuous terms are utilized \cite{WANetal}. 
\item[(c)] The success of the proposed landmark observer design relies on generation of the new scalar LREs \eqref{lre_mixing}, which satisfy the PE condition \cite{ORTetalaut}. It combines two parts, namely, the first contains information in the current small ``interval'' invoking the K-ELRE generated by the filter---the first four equations in \eqref{ldmk-obs}---and the second one consists of historical data using an integral operation. The parameters $\ki^i$ are adopted to play a role of weighting between them. 
    $$
    \begin{aligned}\small
    \underbrace{\vbf{Y}_i(t)}_{\mbox{\footnotesize``current interval''}} + ~\underbrace{ \ki^i \vbf{Y}_i^e(t)}_{\footnotesize\mbox{historical information}} 
~=~
\Delta_i^e(t)
{}^v \ellb_i.
\end{aligned}
    $$ 

	
\end{enumerate}
\end{remark}

\vspace{0.2cm}

\begin{remark}
\label{remark:local_mapping}\rm
The odometry $^v X$ generated from \eqref{dyn_ext1} can be viewed as a \emph{fragile} estimate of the pose $X$. Namely, if an anchor pose is selected in advance, i.e., pre-defining\footnote{
It is standard practice for SLAM algorithms to give an anchor $X(0)$ in advance, since it is widely recognized that the system \eqref{kinematics}-\eqref{dyn:ldmk} and \eqref{output} is not strongly differentially observable \cite{BES}. The underlying reason is clear that it is ambiguous to identify the origin of $\{\cali\}$ with only measurements in the $\{\calb\}$, see for example \cite{ZLOFOR} in which the estimation error converges to a \emph{quotient manifold} rather than an isolated equilibrium. 
} 
$X(0)=X_0:=\calt(R_0,\bfx_0)\in SE(3)$ in $\{\cali\}$, then for \eqref{dyn_ext1} by choosing particular initial conditions we have the invariance
$$\small
  \Big[\xib(0) = \bfx_0, ~ Q(0)= R_0\Big] 
  ~\implies ~
  \Big[ X(t)  = {}^v X(t), ~ \forall t\ge0\Big]
$$
under \emph{ideal circumstance}. This, however, is not practically applicable especially after a long period, since the odometry---the dynamic extension \eqref{dyn_ext1} relies on open-loop integral, which may be a problematic operation yielding \emph{error accumulation}.
\end{remark}


\subsection{Landmark and Pose Observer in the Inertial Frame}
\label{sec42}

In this section, we introduce a localization and mapping observer, which is cascaded to the mapping observer in the dynamic extension frame in Theorem \ref{prop:mapping} in order to express all estimates in the inertial frame, and is robust \emph{vis-\`a-vis} measurement noise. That is to synchronize $\hat X$ to $X$, equivalently estimating the transformation $^c X$.

We now propose the main results of the paper---a localization and mapping observer to fuse the estimates of landmarks and pose together, which provides robust estimation $\hat \bfx, \hat R$ and $\hat \ellb_i$ in the inertial frame.

Our design is motivated by the standard EKF-SLAM \cite{HUADIS}, i.e., making the landmark estimates $\hat \ellb_i$ and the pose estimate $\hat X$ ``coupled'' together.\footnote{In EKF-SLAM, the Jacobian of the output function achieves such coupling between $X$ and $\hat \ellb_i$.} In this way, the observation performance of $\hat \ellb_i$ and $\hat X$ can be \emph{improved from each other}. To this end, we rely on the following technical route:
\begin{itemize}
    \item[-] Estimate the landmarks (denoted as $\bar \ellb_i$) in the inertial frame, using the information of measurement $\bfy,\bfu$, and ${}^v X$, as well as the anchor $(R_0, \bfx_0)$---see Eq. \eqref{bar_l};
    
    \item[-] Use the ``observation error'' between $\bar\ellb_i$ and ${}^v \ellb_i$ to estimate the transformation $^c Q$; see the first equation in \eqref{observer:pose}; 
    
    \item[-] Express all estimates in the inertial frame, with the estimation convergence $\hat X= \calt(\hat R, \hat \bfx) \to X$ and $\hat\ellb_i \to \ellb_i$ as $t\to\infty$; see Theorem \ref{prop:pose-observer}.

\end{itemize}

\vspace{0.2cm}

{\em Localization and mapping observer:} 
\begin{equation}
\label{observer:pose}
\left\{~
\begin{aligned}
{}^c\dot{\hat Q} & ~=~  - ( \bfw_{\tt vis} )_\times {}^c \hat Q \\
\dot{\hat \bfx} &~ = ~\hat R\bfv  + \sum_{i \in \caln} \sigma_i \big[\bar \ellb_i - \hat \bfx - {}^c \hat Q^\top(\hat{\ellb}_i^v - \xib) \big],
\end{aligned}
\right.
\end{equation}
with 
\begequ
\label{bar_l}
\begin{aligned}
\dot\eta_{\Sigma_i} &= F_{\Sigma_i} \Big(\eta_\Sigma, \phi_i^\top(\xib - \xib(0)+ Q(0) R_0^\top \bfx_0),  R_0 Q(0)^\top \phi_i \Big)
\\
\bar\ellb_i & = N_\Sigma(\eta_\Sigma),
\end{aligned}
\endequ
where the functions $F_\Sigma, N_\Sigma$ defined by the system $\Sigma$ in \eqref{ldmk-obs} and
\begin{itemize}
    \item[-] the gains:
    $
    k_i, \sigma_i >0, ~ \forall i;
    $
    
    \item[-] the variables: 
    $$
    \begin{aligned}
    {}^v \hat\bfr_i &:= {}^v \hat \ellb_{i+1} - {}^v \hat \ellb_{i},
    ~~
\vbf{w}_{\tt vis}   := \sum_{i =1}^{n-1} k_i  {}^v\hat{\bfr}_i \times ( {}^c \hat Q \bar \bfr_i )
\\
\bar\bfr_i & := \bar\ellb_{i+1} - \bar\ellb_i
\end{aligned}
$$

    \item[-] the observer outputs:
    \begequ
\label{ldmk-obs-alg}
\begin{aligned}
&\hat R  & =& ~~{}^c\hat{Q}^\top Q
\\
&\hat \ellb_i & =& ~~ {}^c \hat Q^\top ({}^v  \hat \ellb_i - \xib) + \hat \bfx, \quad i \in \caln.
\end{aligned}
\endequ
\qed
\end{itemize}

The properties of the above observer are given as follows.

\vspace{0.2cm}

\begin{theorem}
\label{prop:pose-observer}\rm ({\em Localization and mapping})
Consider the kinematic model \eqref{kinematics} under Assumptions \ref{ass:ie}-\ref{ass:3ldmk} with the given anchor $(R_0, \bfx_0)$. The localization and mapping observer \eqref{bar_l}, \eqref{observer:pose} and \eqref{ldmk-obs-alg}, and the dynamic extension \eqref{dyn_ext1} from the initial condition $\calt(Q(0),\xi(0))$, together with the mapping observer in Theorem \ref{prop:mapping}, achieve the task \eqref{convergence} almost globally.\footnote{The qualifier ``almost'' stems from the fact that the set of initial conditions which do not guarantee the convergence has zero Lebesgue measure.} 
\end{theorem}
\begin{proof}
The proof is given in Appendix.
\end{proof}

\vspace{0.2cm}

The above design provides an approach to deal with drift in odometry. The discretization of the proposed PEBO-SLAM observer design and its implementation procedure is summarized in Algorithm \ref{alg:1}.

\subsection{Discussions}

\begin{algorithm}[!htb]
		\caption{ Visual PEBO-SLAM (discrete time)}
		\begin{algorithmic}[1] \label{alg:1}
			\renewcommand{\algorithmicrequire}{\textbf{Input:}} 
			\renewcommand{\algorithmicensure}{\textbf{Output:}}
			\REQUIRE  Discrete-time IMU measurements (velocities $\bfu(k)$), bearing measurements $\bfy_i(k)$ ($k\in \mathbb{N}$), sampling time $T$ and the final moment $k_{\tt max}$
			\ENSURE  $\hat{X}(k)$ and $\hat{\ellb}_i(k)$ for all $k$ and $i \in \caln$

			\STATE Initialize the observer: $X(0)= \calt(R_0,\bfx_0)$, ${}^c \hat Q(0) = I$ and $\hat \bfx(0)=\xib(0)$.

		    \WHILE{$k<k_{\tt max}$}
			\STATE Compute ${}^v X$ from the dynamic extension
			$$
			\begin{aligned}
			\xib(k+1) & ~=~ \xib(k) + T Q(k)\bfv(k)
			\\
			Q(k+1) & ~= ~Q(k)\exp(T\Omegab_\times)
			\end{aligned}
			$$
			\algocomment{${}^v X= \calt(Q,\xib)$ with $\calt\cal(\cdot)$ defined in \eqref{calt}}
			\STATE Run the gradient estimator to identify ${}^v\hat\ellb_i$

 \IF{the $i$-th landmark is visible}	
\STATE update the estimate ${}^v \hat{\ellb}_i$ as
$$\footnotesize
			\begin{aligned}
             {}^v\hat\ellb_i(k+1) = & {}^v\hat\ellb_i(k) 
              +{1 \over \gamma + 1}
             \left[
             \bfq_i(k+1) - \phi_i(k+1) {}^v\hat\ellb_i^v(k)
             \right]
             \end{aligned}
$$
~~$\phi_i  := \Pi_{Q\bfy_i}^\top$, ~ $\bfq_i  := \Pi_{Q\bfy_i}  \xib$

\algocomment{$\alpha_i >0$ for all $i\in \caln$}

\algocomment{may be replaced by the DREM estimator \eqref{ldmk-obs}}			
	  
\ELSIF{it is invisible}
    \STATE update it as ${}^v \hat{\ellb}_i(k) = {}^v \hat{\ellb}_i(k-1)$
  \ENDIF
	
\STATE Run the localization observer
$$\small
\begin{aligned}
{}^c \hat Q(k+1) & = \exp( -T\cdot S(\bfw_{\tt vis})  ){}^c \hat Q(k) \\
{\hat \bfx}(k+1) 
&=\hat \bfx(k) + T\left[\hat R(k)\bfv (k) + \sum_{i=1}^{n} \sigma_i \big({\hat \bfz_i}(k) - \hat \ellb_i(k) \big)
\right]
\\
\hat R({k}) & = {}^c \hat Q^\top(k) Q(k)
\\
\hat \ellb_i (k) & = {}^c \hat Q^\top(k) ({}^v \hat \ellb_i(k) - \xib(k)) {+} \hat \bfx (k)
\end{aligned}
$$
with the functions
$$
\begin{aligned}
\bfw_{\tt vis} & ~ =~ 
\left\{
\begin{aligned}
&\sum_{i=1}^{n -1} k_i  \hat{\bfr}_i^v {\times} ( {}^c \hat Q \bar \bfr_i ) & \mbox{if~} n >2
\\
&0 & \mbox{otherwise}
\end{aligned}
\right.
\\
{}^v\hat \bfr_i & ~=~ {}^v\hat \ellb_{i+1} - {}^v\hat \ellb_i , \quad
\bar \bfr_i  = \bar \ellb_{i+1} - \bar \ellb_i.
\end{aligned}
$$
\algocomment{for all $i \in \caln$}


\STATE Run the discretized gradient estimator
$$\small
\begin{aligned}
\bar \ellb_i(k+1)  & ~=~ \bar \ellb_i(k) 
+{1 \over \gamma + 1}
\left[ 
\bfq_{x,i}(k+1) - \hat \phi_i(k+1) \bar \ellb_i(k)
\right]
\\
\bfq_{x,i}(k) & ~=~ \hat\phi_i(k) {\hat \bfx}(k)
, \quad
\hat\phi_i(k) = 
\left\{
\begin{aligned}
& \Pi_{\hat R(k)\bfy_i(k)} & & \mbox{visible}
\\
& 0_{3\times 3} & & \mbox{invisible}
\end{aligned}
\right.
\end{aligned}
$$

\STATE
\(
\hat R(k) = {}^c \hat Q^\top(k) Q(k)
\)
\STATE \( \hat \ellb_i(k) = {}^c \hat Q^\top(k) ( {}^v  \hat\ellb_i(k) - \xib(k))  + \hat \bfx(k)\)

\algocomment{for all $i \in \caln$}

\STATE $k = k+1$
\ENDWHILE
\RETURN $\hat{X}= \calt(\hat R, \hat \bfx),~\hat\ellb_i $ for all $k \in \mathbb{N}$ and $i \in \caln$.
\end{algorithmic}
\end{algorithm}

In this section we give some remarks to the proposed design.

\begin{remark}\rm
Note that in the estimator \eqref{bar_l} the predefined anchor is involved, rather than the information of $\hat X$. In the standard SLAM formulation, it is natural to impose the constraint of initializing the estimate from the anchor
 $
 \hat R(0) =R_0, ~\hat \bfx (0) = \bfx_0.
 $
 Then, in \eqref{bar_l} we may replace the ``anchor'' $(R_0, \bfx_0)$ by its estimate $(\hat R(t),\hat \bfx(t))$ at the \emph{current} moment $t$; correspondingly $\xi(0)$ and $Q(0)$ are replaced by $Q(t)$ and $R(t)$. Namely, $\bar \ellb_i$ is generated by\footnote{Roughly speaking, it can be understood that the anchor and the initial moment are defined in a ``moving manner''. The conclusion still holds true since the observer for $\hat X= \calt(\hat R, \hat \bfx)$ is asymptotically stable, by invoking that asymptotic stability implies the invariance $[\hat R(0)= R_0, \hat \bfx(0)=\bfx_0] \implies [\hat R(t)= R(t), \hat \bfx(t)=\bfx(t)]$ for all $t\ge 0$.}
\begequ
\label{bar_l2}
\begin{aligned}
\dot\eta_\Sigma &= F_\Sigma \Big(\eta_\Sigma,\phi_i^\top( Q \hat R^\top \hat\bfx),   \hat R Q^\top \phi_i  \Big)
\\
\bar\ellb_i & = N_\Sigma(\eta_\Sigma),
\end{aligned}
\endequ
 In this way, the landmark estimate $\hat \ellb_i$ in $\{\cali\}$ and the pose estimate $(\hat R, \hat\bfx )$ are in a ``feedback close-loop''.
\end{remark}

\begin{remark}
\rm 
From the fourth item of Theorem \ref{prop:mapping}, we note that the convergence rate is proportional to the gain $\gamma_i$ after the moment $(t_0+t_c+\tau_\star)$. The fact that we need to wait a moment to tune the observer is intuitive, since the $i$-th feature point gets sufficient excitation after $t_0+t_c$. Due to favorable decoupling property, we may easily tune the convergence rate of each element ${}^v\tilde \ell_{i,j}$ separately.
\end{remark}

\begin{remark}
\rm\label{remark:ltvkf}
As illustrated before, the SLAM problem is concerned with two tasks, i.e., mapping and localization. If we are only interested in mapping, Theorem \ref{prop:mapping}, indeed, provides a feasible solution to it. A similar problem has been comprehensively studied in \cite{LOUetal} via transforming the system into an LTV system, and then solved by LTV Kalman filter, in which the pose estimation is not involved, and the landmark coordinates ${}^{\calb} \ellb_i$ in the body frame $\{\calb\}$ are selected as the system state. In that case, the landmark coordinates  $\li^{\calb}$ keep ``rotating'' when the algorithm is running. In contrast, a favorable property of the landmark observer in Theorem \ref{prop:mapping} relies on that the landmark coordinates ${}^v \ellb_i$ are \emph{stationary}.
\end{remark}

\begin{remark}
\rm
Algorithm \ref{alg:1} is applicable to the two-dimensional case with slight modifications. For this case, $\bfv\in \rea^2, \Omegab \in \rea, \bfx\in \rea^2$ and $R\in SO(2)$. The algorithm is modified by defining the operation $\mathbf{a}\times \mathbf{b} = a_1b_2- a_2b_1$ for $\mathbf{a},\mathbf{b} \in \rea^2$, and the operator
$
\Omegab_\times = \begmat{0 & - \Omegab \\ \Omegab & 0}.
$
\end{remark}


\begin{remark}\rm
In stability analysis of the localization observer, we require that at least three feature points visible at all the time to guarantee the asymptotic convergence of the attitude estimate. If this cannot be verified in some intervals, we may also involve the feature points in recent past for the function $\vbf{\omega}_{\tt vis}$, since both $\bfr_i$ and ${}^v \bfr_i$ are constant vectors.
\end{remark}


\begin{remark}\rm
In the proposed observer, we utilize the vectors $\bar\bfr_i$ and their corresponding estimates $\hat\bfr_i$ to generate the error correction term, as illustrated in \eqref{observer:pose}-\eqref{ldmk-obs-alg}. Specially, we define the vectors as $\bar\bfr_i  := \bar\ellb_{i+1} - \bar\ellb_i$, which represent a ``minimal'' construction ensuring convergence conditions as proven in our analysis. However, it is worth noting that alternative approaches can be explored for selecting these vectors among landmarks. One such approach involves considering a ``fully connected graph" among all visible landmarks, which could potentially enhance robustness but at the expense of significantly increased computational burden, growing quadratically. Further investigation into the construction of these vectors among landmarks would be valuable and warrants future research.
\end{remark}

\section{Simulations}
\label{sec5}

In this section, we give some simulation results to demonstrate effectiveness of the proposed observer, considering the cases with and without sufficient excitation.

\subsection{With sufficient excitation}
\label{sec:51}

We first consider a scenario that the robot keeps moving, such that the PE condition is satisfied for all feature points. A robot is simulated from the point $\bfx(0)=[1,1,2]^\top$ and the initial attitude $R(0)= [\cos({\pi\over 6}), -\sin({\pi \over 6}), 0; \sin({\pi \over 6}), \cos({\pi \over 6}), 0 \\ 0 , 0, 1]$
with the velocity $\bfv(t) = [1, 0, 0]^\top$ and $\Omegab(t) = [0,0,-0.4]^\top$; see Figs. \ref{fig:1-1} and \ref{fig:1-3} for the pose trajectory. 
\begin{figure}[!htb]
    \centering
    \subfigure[Robot trajectory $\bfx$ and its estimate]{
    \includegraphics[width=0.235\textwidth]{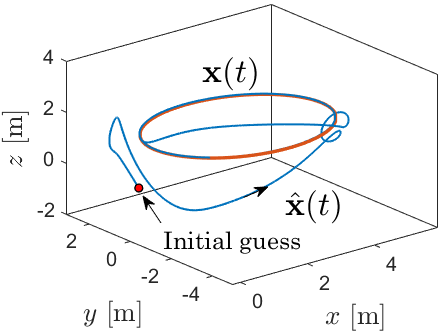}
    \label{fig:1-1}
    }
    \hspace{-0.5cm}
    \hfill
    \subfigure[Pose estimation error $\tilde\bfx$]{
    \includegraphics[width=0.235\textwidth]{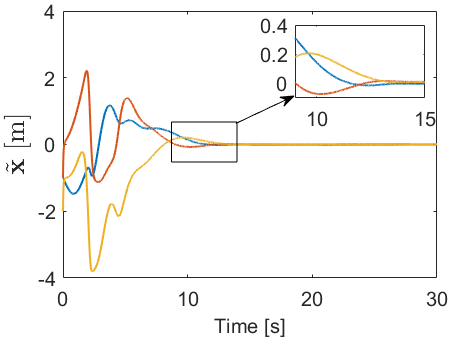}
    \label{fig:1-2}
    }
    \subfigure[Attitude $R$ and its estimate]{~~~~~
    \includegraphics[width=0.17\textwidth]{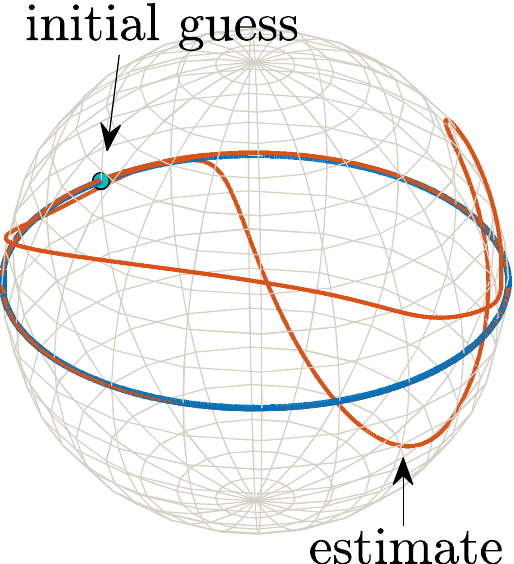}
    \label{fig:1-3}
    }
    \hfill
    \subfigure[Attitude estimation error $\tilde R$]{
    \includegraphics[width=0.235\textwidth]{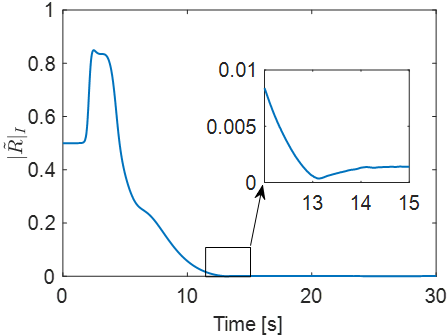}
    \label{fig:1-4}
    }
    \caption{Simulation results with the robot trajectory satisfying PE}
    \label{fig:simulation1}
\end{figure}
We consider six landmarks in simulation, and assume that they are measured all the time. The initial conditions are set as
$$\small
Q(0) = \begmat{\cos({\pi\over 2}) & -\sin({\pi \over 2}) & 0\\ \sin({\pi \over 2}) & \cos({\pi \over 2}) & 0 \\ 0 & 0 & 1},
~
{}^c \hat Q(0) = I_3,
~
\xib(0) = \begmat{0 \\ 1 \\1 },
$$
and $\hat \bfx(0) = 0_{3\times 1}$. The observer gains are selected as $\alpha_i = 5, \gamma _i =100, k_{\tt I}^i = 5$ for $i\in \caln$, and $\rho_j = 1$ for $j=1,2,3$. The simulation was done in Matlab/Simulink with the fixed step size 0.001 seconds. To evaluate robustness of the observer, as well as to make simulations more realistic, we added noise to the measured velocities and the bearing.

First we show the simulation results of the mapping observer in the dynamic extension frame in Fig. \ref{fig:1-5}, verifying the theoretical results in Theorem \ref{prop:mapping}, in particular, the element-wise monotonicity providing favorable performance. This property is also illustrated by Fig. \ref{fig:1-8}, which shows the evolution of each element ${}^v \tilde{\ell}_{i,j}$ the estimation error for the second landmark. Then, we evaluate the performance of the localization observer in Theorem \ref{prop:pose-observer} with the simulation results given in Figs. \ref{fig:simulation1}, which illustrates the robot pose $(R,\bfx)$ and its estimate. 
\begin{figure}[!htb]
    \centering
    \subfigure[Landmark estimates ${}^v\hat\ellb_i$ in $\{\calv\}$]{
    \includegraphics[width=0.235\textwidth]{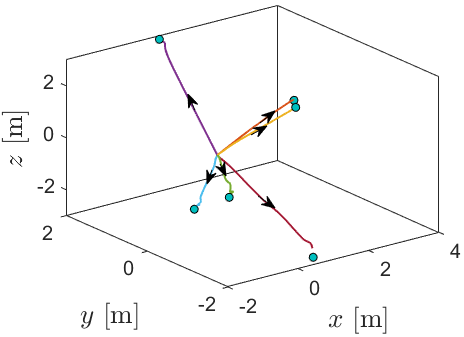}
    \label{fig:1-5}
    }
    \hspace{-0.5cm}
    \hfill
    \subfigure[Landmark estimates $\hat\ellb_i$ in $\{\cali\}$]{
    \includegraphics[width=0.235\textwidth]{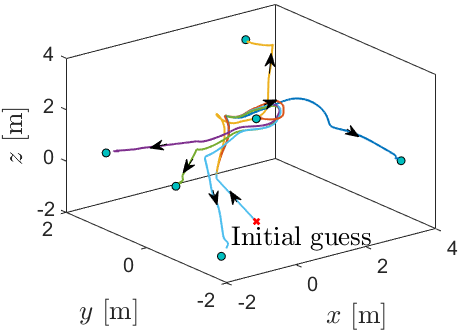}
    \label{fig:1-6}
    }
    \subfigure[Landmark estimation error in $\{\calv\}$]{
    \includegraphics[width=0.235\textwidth]{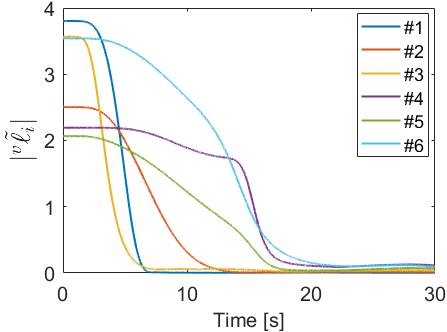}
    \label{fig:1-7}
    }    
    \hspace{-0.5cm}
    \hfill
    \subfigure[Error for the 2nd landmark]{
    \includegraphics[width=0.235\textwidth, height =3.2cm]{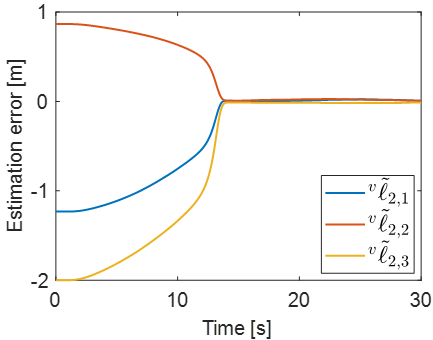}
    \label{fig:1-8}
    }    
    \caption{Feature points estimation under the PE condition}
    \label{fig:simulation2}
\end{figure}
Note that we adopt the metric $|\tilde R|_I $ with $\tilde R:= R \hat R^\top$ to characterize the attitude estimation error with the definition $|\tilde R|_I^2 := {1\over4} \tr(I_3 - \tilde R)$. For a better illustration, we visualize in Fig. \ref{fig:1-3} presenting the attitude trajectory and its estimate by means of the unit vector along $x$-axis, i.e., $e_1 = [1, 0, 0]^\top$, of the body-fixed frame, which is plotted in the inertial frame. Note that the attitude can be uniquely determined in terms of right-handed coordinate. As expected, after a short transient stage the pose estimate asymptotically converges to its true value with a small ultimate error, which is caused by measurement noise. It also shows the robustness of the landmark observer w.r.t. high-frequency measurement noise. The estimate of the landmark coordinates in the inertial frame $\{\cali\}$ are given in Fig. \ref{fig:1-6}. 
\begin{figure}[!htb]
    \centering
    \subfigure[The effect to $|{}^v \tilde{\vbf{\ell}}|$]{
    \includegraphics[width=0.235\textwidth]{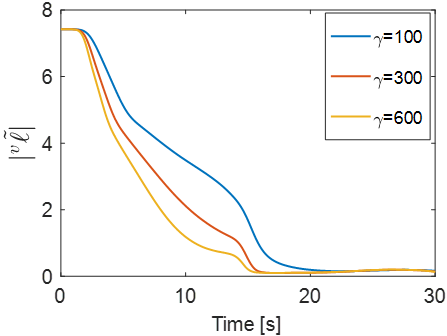}  \label{fig:s9}
    }
    \hspace{-0.5cm}
    \hfill
    \subfigure[The effect to $\tilde R$]{
    \includegraphics[width=0.235\textwidth]{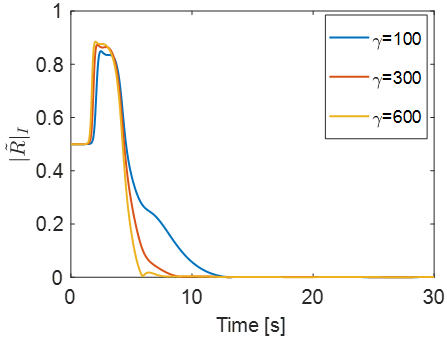}  \label{fig:s8}
    }
    \caption{The effect of parameter $\gamma$ to the convergence rate}
    \label{fig:parameter}
\end{figure}
At the end, we show the effect of the gain parameters $\gamma_i$ to the convergence rates for the estimates ${}^v\hat{\ellb}$ and $\hat R$. In Fig. \ref{fig:s9}, it verifies the theoretical analysis in the fourth term of Theorem \ref{prop:mapping}, {i.e.}, a larger gain implying a faster convergence rate. However, it cannot be arbitrarily fast since we need to wait a while to obtain sufficient excitation. On the other hand, the estimation errors ${}^v \tilde\ellb_i$ appear in the dynamics of $\tilde R$, and thus tuning $\gamma_i$ also affects the convergence rate of $\tilde R$.


\subsection{Not satisfying the PE condition}

To illustrate the merits of the proposed design, in this section we consider robot trajectories without sufficient excitation. Specifically, all the scenarios are adopted from that in Section \ref{sec:41} but with a different trajectory of the robot. In this case, we assume that the robot stopped at $t=12$s with velocity
$$\footnotesize
\bfv=
\left\{
\begin{aligned}
{}	&[1,0,0]^\top & t\in [0,12]\\
	&0_{3\times1}, &  t\ge 12
\end{aligned}
\right.
,
\quad
\Omegab = \left\{
\begin{aligned}
{} & [0,0,-0.4]^\top, &t \in [0,12]\\
& 0_{3\times 1}, & t\ge 12 	
\end{aligned} \right.
.
$$
It guarantees that all the landmarks satisfy the IE condition with $t_0=0$ and $t_c=12$s, but not guaranteeing PE. The observer gains are selected as $\alpha_i = 5, \gamma _i =100, k_{\tt I}^i = 20$ for $i\in \caln$, and $\rho_j = 1$ for $j=1,2,3$. From Figs. \ref{fig:simulation2}-\ref{fig:simulation-ldmk} we observe that the pose estimation has a satisfactory performance, even when the PE condition are not verified.  The pose trajectory is given in Figs. \ref{fig:2-1} and \ref{fig:2-3}. We underline that this scenario can be considered mathematically equivalent to the situation where the landmarks become invisible after 12 seconds, after which the excitation condition does not hold anymore. Indeed, this is more common in various SLAM problems. 

\begin{figure}[!htb]
    \centering
    \subfigure[Robot trajectory $\hat \bfx$ and its estimate]{
    \includegraphics[width=0.235\textwidth]{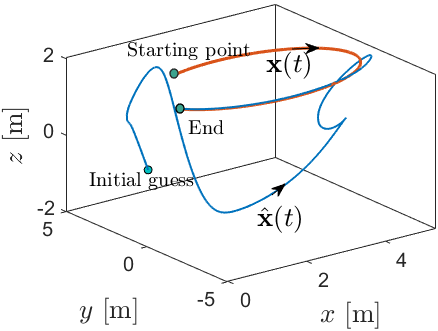}
    \label{fig:2-1}
    }
    \hspace{-0.5cm}
    \hfill
    \subfigure[Pose estimation error $\tilde\bfx$]{
    \includegraphics[width=0.235\textwidth]{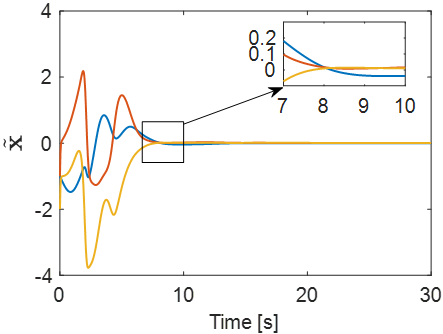}
    \label{fig:2-2}
    }
    \subfigure[Attitude $R$ and its estimate]{~~~
    \includegraphics[width=0.19\textwidth]{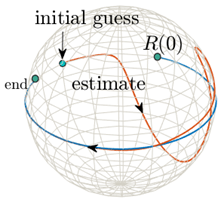}
    \label{fig:2-3}
    }
    \hfill
    \subfigure[Attitude estimation error $\tilde R$]{
    \includegraphics[width=0.235\textwidth]{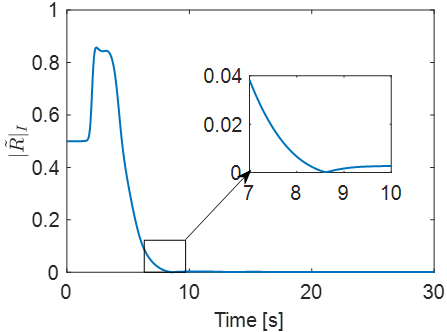}
    \label{fig:2-4}
    }
    \caption{Simulation results with the robot trajectory not satisfying PE}
    \label{fig:simulation2}
\end{figure}
\begin{figure}[!htb]
   \centering
   \subfigure[${}^v\hat{\ellb}_i$ in the dynamic extension frame]{
   \includegraphics[width=0.23\textwidth]{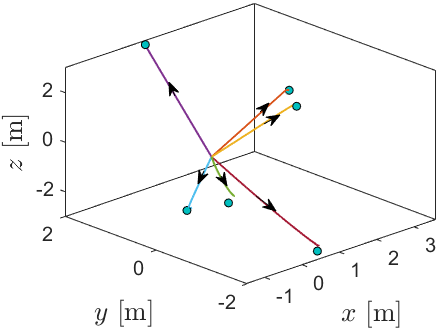}
   \label{fig:c1}
   }
   \hspace{-0.5cm}
   \hfill
   \subfigure[$\hat{\ellb}_i$ in the inertial frame]{
   \includegraphics[width=0.23\textwidth]{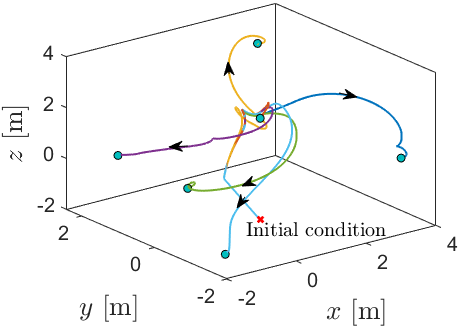}
   \label{fig:c2}
   }
    \caption{Landmark estimates in the dynamic extension and the inertial frames without PE }
    \label{fig:simulation-ldmk}
\end{figure}

We finally compare the proposed design to the robo-centric SLAM observer in \cite{LOUetal}, which requires the system trajectory being UCO. Indeed, the observer in \cite{LOUetal} studies the case that the feature coordinates are expressed in the body-fixed frame $\{\calb\}$, i.e.,
$
{}^{\calb}\ellb_i = R^\top \left( \ellb_i  - \bfx \right) ,
$
which are time-varying signals. In order to make a fair comparison, we plot the evolution of the \emph{norms} of the observation errors, i.e., $|\tilde \ellb_i|$ of the proposed design and $|{}^{\calb}\tilde{\ellb}_i|$ for the one in \cite{LOUetal}, since rotation does not affect norms. As shown in Fig. \ref{fig:comparison}, the proposed observer guarantees the estimates converging to some relatively small neighbourhoods of their true values in the presence of only the IE condition, showing good robustness against noise. In contrast, Fig. \ref{fig:c8} shows that the estimates from the observer in \cite{LOUetal} stop converging at the moment $t_0+t_c=12$s with large steady-state errors, since the UCO condition required in LTV Kalman filter cannot be satisfied. It is interesting to note that the estimates diverge from that moment due to the accumulation of noise, which is conspicuous by its absence in our design.


\begin{figure}[!htb]
   \centering
   \subfigure[Landmark estimation errors $|\tilde{\ellb}_i|$ of the proposed design in the inertial frame]{
   \includegraphics[width=0.215\textwidth]{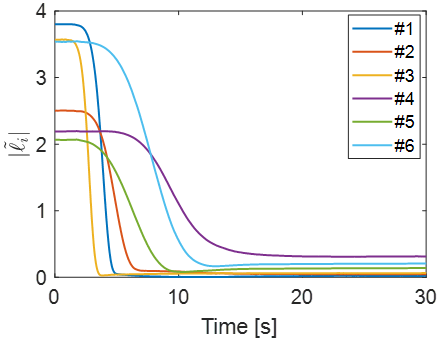}
   \label{fig:c7}
   }
   ~
   \hfill
    \subfigure[Landmark estimation errors $|{}^{\calb}\tilde \ellb_i|$ of the LTV Kalman filter \cite{LOUetal} in the body-fixed frame]{
   \includegraphics[width=0.215\textwidth]{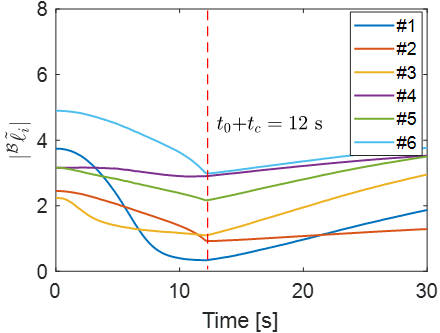}
   \label{fig:c8}
   }
   \caption{Comparison between the proposed observer and the LTV Kalman filter in \cite{LOUetal}}
    \label{fig:comparison}
\end{figure}

{~}

\section{Conclusion}
\label{sec6}

In this paper we have introduced a novel visual inertial SLAM observer framework called PEBO-SLAM. A key observation is that the coordinates of feature points in the dynamic extension frame are constant, based on which we are able to get a set of linear regression models, and then transform the problem into online parameter estimation. To address this, we extend the PEBO methodology to the manifold $SE(3)\times \rea^{3n}$, the unknown ``parameters'' in our context being the landmark coordinates ${}^v \ellb_i$ together with the constant relative rigid transformation ${}^c X$. A simple constructive solution is also provided, with guaranteed almost global asymptotic stability, while significantly relaxing the strong PE or UCO-type conditions required in the existing literature. As future work, it is of practical interests to study the performance limitation from gyro noise and bias, as well as how to robustify the proposed observer. The final remark is that, as a byproduct, the proposed method can be also used as an initialization algorithm due to the globality and convexity.

\section{Acknowledgement}
Partial work has been done when the second author was with I3S-CNRS, France. He appreciates the fruitful discussions with Prof. Tarek Hamel.

\appendix

\subsection{Proof of Theorem \ref{prop:pose-observer}}

The proof consists of two parts: 1) illustrating the exponential convergence of $|\bar \ellb_i - \ellb_i| \to 0$ as $t\to\infty$, which will be used to synchronize $\hat X$ to $X$; and 2) showing that $\hat \ell_i$, $\hat \bfx$ and $\hat R$ converges to their true values in the inertial frame.

{\em Part 1}. Since the anchor is fixed at $X(0)=\calt(R_0,\bfx_0)$, recalling \eqref{lre1} it is easy to get the LREs
\begin{equation}
\begin{aligned}
 & \quad \vbf{q}_i =  \Pi_{Q\bfy_i} {}^v\ellb_i
 \\
  \overset{\eqref{const:ldmk}}\implies & \quad \vbf{q}_i =  \Pi_{Q\bfy_i} \Big[ {}^c\xib + {}^cQ \ellb_i \Big], 
 \\
 \overset{\eqref{id:1}}\implies & \quad
 \vbf{q}_i = \Pi_{Q\bfy_i} \Big[ \xib(0) - Q(0)R_0^\top \bfx_0 + Q(0) R_0^\top \ellb_i \Big].
 \end{aligned}
\end{equation}
Then, we have the linear regression model
$$
\phi_i^\top(\xib - \xib(0)+ Q(0) R_0^\top \bfx_0) =  [ R_0 Q(0)^\top \phi_i]^\top \ellb_i.
$$
It is straightforward to verify 
$$
\phi_i \in \mbox{~PE} \quad \implies \quad [R_0 Q(0)^\top \phi_i]\in \mbox{~PE}
$$
by recalling the full-rankness of $R_0$ and $Q(0)$. Using the same arguments in Theorem \ref{prop:mapping}, we able to show that for $i \in \caln$
\begequ
\label{converge:barl2l}
\lim_{t\to\infty} |\bar\ellb_i(t) - \ellb_i| = 0 \quad \mbox{(exp.)}.
\endequ

Hence, the variables ($\bar \ellb_i$)s can be viewed as ``known'' vectors in $\{\cali\}$. Then, we may generate an error between $\bar \ellb_i$ and ${}^v\hat\ellb_i$ to synchronize $\hat X$ to $X$ asymptotically, and thus we are able to get the estimates of all landmarks in the inertial frame. To be precise, we need to show the asymptotic convergence of ${}^c \hat X= \calt({}^c \hat Q,{}^c \hat \xib) \in SE(3)$ to its true value ${}^c X$ defined in Lemma \ref{lem:1}, which is sufficient to show $\hat X \to X$ in terms of the identity \eqref{id:1} for the proposed PEBO.

{\em Part 2}. Define the estimation error of ${}^cQ$ as 
$$
{}^c\tilde Q := {}^c Q {}^c \hat Q^\top,
$$
and then 
\begequ
\label{tilde_qc}
\begin{aligned}
	{}^c \dot{\tilde Q}
	= {}^c Q {}^c \dot{\hat Q}^\top 
	= {}^c \tilde Q (\vbf{w}_{\tt vis})_\times.
\end{aligned}
\endequ
We also define a \emph{nominal} correction term
\begin{equation}
\label{w*}
\vbf{w}_{\tt vis}^* = \sum_{i =1}^{n-1} k_i \bfr_i^v \times ({}^c \hat Q \bfr_i),
\end{equation}
and we will show that the difference $|\bfw_{\tt vis}^* -\bfw_{\tt vis}|\to 0$ exponentially as $t\to 0$. It is easy to see
$$
\begin{aligned}
(\bfw_{\tt vis}^*)_\times & ~=~ \hal \sum_{i =1}^{n-1} k_i \left[ {}^c \hat Q \bfr_i ({}^v \bfr_i)^\top - {}^v \bfr_i ({}^c \hat Q \bfr_i)^\top \right]
\\
& ~=~ \sum_{i =1}^{n-1} k_i \pa \left( {}^c \hat Q \bfr_i \bfr_i^\top {}^cQ^\top \right)
\\
& ~=~ \sum_{i =1}^{n-1} k_i \pa \left({}^c\tilde{Q}^\top {}^cQ \bfr_i \bfr_i^\top {}^cQ^\top \right)
\\
& ~=~ \pa \left( {}^c\tilde Q^\top M \right),
\end{aligned}
$$
with $M:= \sum_i^{n -1} k_i {}^cQ \bfr_i {\bfr_i^\top} {}^cQ^\top$, and in the second equation we have used
$$
\begin{aligned}
{}^v\bfr_i 
& = {}^v\ellb_{i+1} - {}^v\ellb_i
\\&
 = {}^c\xib  + {}^cQ \ellb_{i+1} - ({}^c\xib + {}^cQ \ellb_i) 
\\& 
= {}^c Q \bfr_i.
\end{aligned}
$$
We underscore that $M$ is a constant symmetric matrix, which is positive semidefinite for $n_\ell =3$, and is positive definite for $n_\ell >3$ due to Assumption \ref{ass:3ldmk}.

On the other hand, since\footnote{With a slight abuse of notation, we use $\et$ to represent some exponentially decaying {\em generally}, and do not distinguish which specific signals they refer to in order to simplify presentation.}
$$
\begin{aligned}
{}^v\hat \bfr_i 
& 
~=~ {}^v\ellb_{i+1} - {}^v \ellb_i  + {}^v\tilde{\ellb}_{i+1}(t) - {}^v\tilde{\ellb}_i(t)
\\ &
 ~=~ {}^v \bfr_i + \et
\end{aligned}
$$
and $\bar \bfr_i = \bfr_i + \et$, invoking the compactness on $SO(3)$ we have
\begin{equation}
\label{w&w*}
\bfw_{\tt vis} = \bfw_{\tt vis}^* + \et.
\end{equation}

From \eqref{tilde_qc} and \eqref{w&w*}, we have
\begequ
\label{tilde_qc2}
{}^c\dot{\tilde Q} = {}^c \tilde Q \pa\left( {}^c \tilde{Q}^\top M  \right) + \cale(t) =: f_q({}^c \tilde  Q, t),
\endequ
in which $\cale(t)$ represents some exponentially decaying signal, invoking compactness of the space which ${}^c\tilde Q$ lives in. Then, we have
$
|\tr(\cale^\top (t)M)| \le a_0 e^{-a_1 t} =: \varepsilon(t),
$
for some $a_0,a_1>0$. Similarly to the proof of \cite[Prop. 6]{YIetalAUT}, we consider the time-varying Lyapunov function candidate
$$
V({}^c \tilde Q,t) = \sum_{j=1}^{n_\ell -1} k_i |\bfr_i|^2 - \tr({}^c \tilde Q^\top M) +  \int_t^\infty | \varepsilon (s)| ds,
$$
which is well posed since $\varepsilon$ is absolutely integrable. Its time derivative satisfies $\dot V \le - \big\| \pa \big({}^c \tilde Q^\top M \big) \big\|^2$.
Then, it yields
$
\int_0^\infty \big\| \pa \big({}^c\tilde Q^\top(t) M \big) \big\|^2 dt < +\infty.
$
Using Babalat's lemma, we conclude that all the trajectories converge to the invariant set
$
\Omega_e := \{{}^c \tilde Q \in SO(3) ~|~ \pa \big({}^c \tilde Q^\top M \big) =0\},
$
which can be shown containing a locally exponentially stable equilibrium ${}^c \tilde Q = I_3$, and other three isolated unstable equilibria $\bq_i$ ($i=1,2,3$) via a similar procedure of the proof in \cite[Thm. 5.1]{MAHetal}.
From the non-autonomous version of Hartman-Grobman theorem \cite{AULWAN}, the dynamics \eqref{tilde_qc2} is topologically equivalent to an LTV dynamics in a small neighborhood of these three unstable equilibria. As a result, only some very specific trajectories, from a zero Lebesgue measure set $\calm_\epsilon$, ultimately converge to the unstable equilibria, thus yielding the almost GAS of the error dynamics \eqref{tilde_qc2}. As a result, 
$
\lim_{t\to\infty}\|\hat R(t) - R(t)\| = 0
$
with $\hat R \in SO(3)$.

After verifying convergence of the attitude estimate $\hat R(t)$, we then need to study the position estimate $\hat \bfx$. To address this, let us define the error vector of $\bfx$ as $\tilde \bfx:= \hat \bfx - \bfx$, and then 
$$
\begin{aligned}
\dot{\tilde \bfx} & ~= ~(\hat R - R) \bfv + \sum_{i =1}^{n-1} \sigma_i \big(\bar \ellb_i - \hat \bfx - {}^c \hat Q^\top(\ell_i^v - \xib) + \et \big)
\\
& ~= ~ (\hat R - R) \bfv  + \sum_{i =1}^{n-1} \sigma_i \big(\ellb_i - \hat \bfx  - {}^c \hat Q^\top Q R^\top (\ellb_i - \bfx) + \et \big)
\\
& ~=~  (\hat R - R) \bfv  + \sum_{i =1}^{n-1} \sigma_i \big(\ellb_i - \hat \bfx  - \tilde R^\top (\ellb_i - \bfx) +\et \big)
\\
& ~=~ - \sum_{i =1}^{n-1} \sigma_i \tilde \bfx +  (\hat R - R) \bfv + 
 \sum_{i =1}^{n-1} \sigma_i \big((I-\tilde R^\top)(\ellb_i - \hat \bfx) \big)
 \\
& ~:=~ - \sum_{i =1}^{n-1} \sigma_i \tilde \bfx + \Delta_{\tt p}(t).
\end{aligned}
$$
From $\|\hat R(t) - R(t)\| \to 0$ and the boundedness of $\bfv(t)$, we have the asymptotic convergence condition $|\Delta_{\tt p}|\to 0$. The dynamics of $\tilde \bfx$ can be viewed as a stable LTI system perturbed by a vanishing term $\Delta_{\tt p}$, and thus following the standard perturbation analysis on Euclidean space we have
$$
\lim_{t\to \infty} |\hat \bfx(t) - \bfx(t) | = 0.
$$

At the end, we show the landmark estimation $\hat \ellb_i$ will converge to their true values ultimately. Invoking the facts $\xib(t) = {}^c\xib + {}^c Q \bfx(t)$, we have
$$
\begin{aligned}
 \li & ~=~ {}^c Q^\top ( {}^v \ellb_i - {}^c\xib) 
\\
& ~=~ {}^c Q^\top ({}^v \ellb_i - \xib + {}^c Q  \bfx)
\\
& ~=~ {}^c Q^\top( {}^v \ellb_i - \xib) + \bfx.
\end{aligned}
$$
Therefore, the landmark estimates $\hat \ellb_i$ in \eqref{ldmk-obs-alg} guarantee
$
\lim_{t\to\infty} |\hat \ellb_i(t) -  \ellb_i| =0.
$
It completes the proof.
\QED

\bibliographystyle{abbrv}
\bibliography{vslam_obs}

\end{document}